\documentclass[11pt]{article}%
\usepackage{amsfonts}
\usepackage{amsmath}
\usepackage{amssymb}
\usepackage{hyperref}
\usepackage{fullpage}
\usepackage{graphicx}%
\providecommand{\U}[1]{\protect\rule{.1in}{.1in}}
\newtheorem{theorem}{Theorem}

\newtheorem{corollary}[theorem]{Corollary}

\newtheorem{definition}[theorem]{Definition}

\newtheorem{lemma}[theorem]{Lemma}

\newtheorem{remark}[theorem]{Remark}

\newenvironment{proof}[1][Proof]{\noindent\textbf{#1.} }{\ \rule{0.5em}{0.5em}}

\let\originalleft\left
\let\originalright\right
\renewcommand{\left}{\mathopen{}\mathclose\bgroup\originalleft}
\renewcommand{\right}{\aftergroup\egroup\originalright}

\begin{document}

\title{\textbf{The squashed entanglement of a quantum channel}}
\author{Masahiro Takeoka\thanks{National Institute of Information and Communications
Technology, Koganei, Tokyo 184-8795, Japan} \thanks{Quantum
Information Processing Group, Raytheon BBN Technologies, Cambridge, MA 02138, USA}
\and Saikat Guha\footnotemark[2]
\and Mark M.~Wilde\thanks{Hearne Institute for Theoretical Physics, Department of
Physics and Astronomy, Center for Computation and Technology, Louisiana State
University, Baton Rouge, Louisiana 70803, USA}}
\maketitle

\begin{abstract}
This paper defines the squashed entanglement of a quantum channel as the
maximum squashed entanglement that can be registered by a sender and receiver
at the input and output of a quantum channel, respectively. A new
subadditivity inequality for the original squashed entanglement measure of
Christandl and Winter leads to the conclusion that the squashed entanglement
of a quantum channel is an additive function of a tensor product of any two
quantum channels. More importantly, this new subadditivity inequality, along
with prior results of Christandl, Winter, \textit{et al}., establishes the
squashed entanglement of a quantum channel as an upper bound on the quantum
communication capacity of any channel assisted by unlimited forward and
backward classical communication. A similar proof establishes this quantity as
an upper bound on the private capacity of a quantum channel assisted by
unlimited forward and backward public classical communication. This latter
result is relevant as a limitation on rates achievable in quantum key
distribution. As an important application, we determine that these capacities
can never exceed $\log\left(  \left(  1+\eta\right)  /\left(  1-\eta\right)
\right)  $ for a pure-loss bosonic channel for which a fraction $\eta$ of the
input photons make it to the output on average. The best known lower bound on
these capacities is equal to $\log\left(  1 /\left(  1-\eta\right)  \right)
$. Thus, in the high-loss regime for which $\eta\ll1$, this new upper bound
demonstrates that the protocols corresponding to the above lower bound are
nearly optimal.

\end{abstract}

\section{Introduction}

One of the seminal insights of classical information theory is that public
discussion between two parties trying to communicate privately can enhance
their ability to do so \cite{M93,AC93}. Indeed, let $p_{Y,Z|X}\left(
y,z|x\right)  $ be a stochastic map modeling a broadcast channel that connects
a sender $X$ to a legitimate receiver $Y$ and a wiretapper $Z$. Maurer
\cite{M93} and Ahslwede and Csisz{\'{a}}r \cite{AC93} independently discovered
that the so-called secret-key agreement capacity of such a broadcast
channel,\footnote{Note that the secret-key agreement capacity is equal to the
capacity for private communication with unlimited public discussion, due to
the one-time pad protocol.} in which public discussion is allowed, can be
strictly larger than zero even if the private capacity of the channel is equal
to zero (the capacity for private communication without any public
discussion). This result has shaped the formulation of practical protocols for
secret key agreement.

In later work, Maurer and Wolf introduced the \textit{intrinsic information}
and proved that it is a sharp upper bound on the secret key agreement capacity
\cite{MW99}. It is defined as follows:%
\begin{equation}
\max_{p_{X}\left(  x\right)  }I\left(  X;Y\downarrow Z\right)  ,
\label{eq:intrinsic-info}%
\end{equation}
where $I\left(  X;Y\downarrow Z\right)  $ is equal to a minimization of the
conditional mutual information over all stochastic maps with which an
adversary possessing $Z$ can act to produce $\bar{Z}$:
\begin{equation}
I\left(  X;Y\downarrow Z\right)  \equiv\min_{p_{\bar{Z}|Z}\left(  \bar
{z}|z\right)  }I\left(  X;Y|\bar{Z}\right)  .
\end{equation}
The conditional mutual information is defined as%
\[
I\left(  X;Y|\bar{Z}\right)  \equiv H\left(  X\bar{Z}\right)  +H\left(
Y\bar{Z}\right)  -H\left(  XY\bar{Z}\right)  -H\left(  \bar{Z}\right)  ,
\]
where the Shannon entropies are evaluated with respect to the marginal
distributions resulting from the joint distribution $p_{X}\left(  x\right)
p_{Y,Z|X}\left(  y,z|x\right)  p_{\bar{Z}|Z}\left(  \bar{z}|z\right)  $. The
interpretation of the intrinsic information is that it is a measure of the
correlations that the sender and legitimate receiver can establish, with the
adversary acting in the strongest possible way to reduce these correlations.

Due to strong parallels discovered between secrecy and quantum coherence or
entanglement \cite{SW98,LC99,SP00}, Christandl realized that an extension of
the intrinsic information to quantum information theory might be helpful in
simplifying the arduous task of quantifying entanglement present in quantum
states \cite{C02}. This realization then culminated in the establishment of
the \textit{squashed entanglement} $E_{\text{sq}}\left(  A;B\right)  _{\rho}%
$\ of a bipartite quantum state $\rho_{AB}$\ as an upper bound on the rate at
which two parties can distill Bell states $\left(  \left\vert 0\right\rangle
\left\vert 0\right\rangle +\left\vert 1\right\rangle \left\vert 1\right\rangle
\right)  /\sqrt{2}$ from many copies of $\rho_{AB}$ by performing local
operations and classical communication \cite{CW04}. A similar proof technique
establishes squashed entanglement as an upper bound on distillable secret key
\cite{CEHHOR07}, and this approach has the benefit of being conceptually
simpler than the original classical approaches from \cite{M93,AC93}. The
squashed entanglement $E_{\text{sq}}\left(  A;B\right)  _{\rho}$ is defined as
the following function of a bipartite state $\rho_{AB}$:
\begin{equation}
E_{\text{sq}}\left(  A;B\right)  _{\rho}\equiv\tfrac{1}{2}\inf_{\mathcal{S}%
_{E\rightarrow E^{\prime}}}I\left(  A;B|E^{\prime}\right)  ,
\label{eq:squashed-ent-state}%
\end{equation}
where the conditional quantum mutual information is defined similarly to the
classical one (however with Shannon entropies replaced by von Neumann
entropies) and the infimum is with respect to all noisy \textquotedblleft
squashing channels\textquotedblright\ $\mathcal{S}_{E\rightarrow E^{\prime}}$
taking the $E$ system of a purification $\left\vert \phi^{\rho}\right\rangle
_{ABE}$\ of $\rho_{AB}$ to a system $E^{\prime}$ of arbitrary dimension. In
related work, Tucci has defined a functional bearing some similarities to
squashed entanglement \cite{T99,T02}.

The similarities between (\ref{eq:squashed-ent-state}) and
(\ref{eq:intrinsic-info}) are evident and the interpretations are similar.
That is, we interpret $E_{\text{sq}}\left(  A;B\right)  _{\rho}$ as
quantifying the quantum correlations between $A$ and $B$ after an adversary
possessing the purifying system $E$ performs a quantum channel with the intent
of \textquotedblleft squashing down\textquotedblright\ the correlations that
they share. Due to a lack of an upper bound on the dimension of the output
$E^{\prime}$\ of the squashing channel, it is not yet known whether the
infimization in (\ref{eq:squashed-ent-state}) can be replaced by a
minimization, but note that such a replacement is possible for the classical
intrinsic information \cite{CRW03}. Furthermore, it is not even clear that,
given a description of a density matrix $\rho_{AB}$, the computation of its
squashed entanglement can be performed in non-deterministic polynomial time
(NP). However, these apparent difficulties are not an obstruction to finding
good upper bounds on the distillable entanglement or distillable secret key of
a quantum state $\rho_{AB}$: the results of \cite{CW04,CEHHOR07} are that any
squashing channel gives an upper bound on these quantities and so the task is
to find the best one in a reasonable amount of time. Finally, among the many
entanglement measures, squashed entanglement is the only one known to satisfy
all eight desirable properties that have arisen in the axiomatization of
entanglement theory \cite{CW04,KW04,AF04,BCY11}.

\section{Summary of results}

In this paper, we provide the following contributions:

\begin{enumerate}
\item Our first contribution is to define the squashed entanglement of a
quantum channel $\mathcal{N}_{A^{\prime}\rightarrow B}$ as the maximum
squashed entanglement that can be registered between a sender and receiver who
have access to the input $A^{\prime}$\ and output $B$\ of this channel,
respectively:%
\[
E_{\text{sq}}\left(  \mathcal{N}\right)  \equiv\max_{\left\vert \phi
\right\rangle _{AA^{\prime}}}E_{\text{sq}}\left(  A;B\right)  _{\rho},
\]
where $\rho_{AB}\equiv\mathcal{N}_{A^{\prime}\rightarrow B}(\left\vert
\phi\right\rangle \left\langle \phi\right\vert _{AA^{\prime}})$. The formula
above is formally analogous to the classical formula in
(\ref{eq:intrinsic-info}), and the rest of this paper establishes that
$E_{\text{sq}}\left(  \mathcal{N}\right)  $ plays an analogous operational
role in the quantum setting.

\item One of the main technical contributions of this paper is a proof of a
new subadditivity inequality (Theorem~\ref{thm:subadd}) for the squashed entanglement.

\item This inequality has two important implications. First,
Theorem~\ref{thm:subadd} implies that $E_{\text{sq}}\left(  \mathcal{N}%
\right)  $ is additive as a function of channels, in the sense that
%the following equality holds%
%\[
$E_{\text{sq}}\left(  \mathcal{N}\otimes\mathcal{M}\right)  =E_{\text{sq}%
}\left(  \mathcal{N}\right)  +E_{\text{sq}}\left(  \mathcal{M}\right)  $ for
any two channels $\mathcal{N}$ and $\mathcal{M}$.
%\]
Thus, the squashed entanglement of a channel is a well behaved function of channels.

\item Next, and more importantly, Theorem~\ref{thm:subadd} is helpful in
establishing $E_{\text{sq}}\left(  \mathcal{N}\right)  $ as an upper bound on
the quantum communication capacity of a channel $\mathcal{N}$ assisted by
unlimited forward and backward classical communication (hereafter denoted as
$Q_{2}(\mathcal{N})$). This new squashed entanglement upper bound is an
improvement upon the best previously known upper bound on $Q_{2}(\mathcal{N})$
given in terms of the entanglement cost of a quantum channel \cite{BBCW13},
with the improvement following from the fact that the squashed entanglement is
never larger than the entanglement of formation \cite{CW04}. In addition to
being tighter, our bound is \textquotedblleft single-letter,\textquotedblright%
\ meaning that it can be evaluated as a function of a single channel use,
whereas the bound from \cite{BBCW13} is regularized, meaning that it is
intractable to evaluate it in spite of being able to write down a formal
mathematical expression for it. By a similar proof, we find that
$E_{\text{sq}}\left(  \mathcal{N}\right)  $ is a single-letter upper bound on
the private capacity of a channel $\mathcal{N}$ assisted by unlimited forward
and backward public classical communication (hereafter denoted as
$P_{2}(\mathcal{N})$). These latter results represent important progress on
one of the longest standing open questions in quantum information theory
\cite{BDS97} (namely, to determine these capacities or sharpen the bounds on them).

\item As examples, we compute upper bounds on $Q_{2}$ and $P_{2}$ for all
qubit Pauli channels, and we consider the special cases of a qubit dephasing
channel and a qubit depolarizing channel.

\item Finally, we show that our upper bound on $Q_{2}({\mathcal{N}}_{\eta})$
and $P_{2}({\mathcal{N}}_{\eta})$ for the pure-loss bosonic channel
${\mathcal{N}}_{\eta}$ with transimissivity $\eta\in\left[  0,1\right]  $, is
very close to the best-known lower bound on these capacities from
\cite{GPLS09,PGBL09}, in the practically-relevant regime of high loss
($\eta\ll1$). This result puts an upper limit on the secret-key rate
achievable by any optical quantum key distribution protocol. We also establish
an upper bound on $Q_{2}$ and $P_{2}$ for all phase-insensitive Gaussian
channels, which includes the thermal and additive noise Gaussian channels.
\end{enumerate}

\section{Properties of the squashed entanglement of a quantum channel}

We begin with our main definition:

\begin{definition}
The squashed entanglement of a quantum channel $\mathcal{N}_{A^{\prime
}\rightarrow B}$ is the maximum squashed entanglement that can be registered
between a sender and receiver who have access to the input $A^{\prime}$\ and
output $B$\ of this channel, respectively:%
\begin{equation}
E_{\operatorname{sq}}\left(  \mathcal{N}\right)  \equiv\max_{\left\vert
\phi\right\rangle _{AA^{\prime}}}E_{\operatorname{sq}}\left(  A;B\right)
_{\rho}, \label{eq:squashed-ent-channel}%
\end{equation}
where $\rho_{AB}\equiv\mathcal{N}_{A^{\prime}\rightarrow B}(\left\vert
\phi\right\rangle \left\langle \phi\right\vert _{AA^{\prime}})$.
\end{definition}

\begin{remark}
We can restrict the optimization in (\ref{eq:squashed-ent-channel}) to be
taken over pure bipartite states rather than mixed ones, due to the convexity
of squashed entanglement (see Proposition~3 of \cite{CW04}). In more detail,
let $\sigma_{AA^{\prime}}$ be a mixed state on systems $A$ and $A^{\prime}$.
Then it has a spectral decomposition of the following form:%
\[
\sigma_{AA^{\prime}}=\sum_{x}p_{X}\left(  x\right)  \left\vert \psi
_{x}\right\rangle \left\langle \psi_{x}\right\vert _{AA^{\prime}}.
\]
Let $\omega_{AB}\equiv\mathcal{N}_{A^{\prime}\rightarrow B}(\sigma
_{AA^{\prime}})$ and $\omega_{AB}^{x}\equiv\mathcal{N}_{A^{\prime}\rightarrow
B}(\left\vert \psi_{x}\right\rangle \left\langle \psi_{x}\right\vert
_{AA^{\prime}})$. Then the following inequality holds, due to convexity of the
squashed entanglement:%
\[
E_{\operatorname{sq}}\left(  A;B\right)  _{\omega}\leq\sum_{x}p_{X}\left(
x\right)  E_{\operatorname{sq}}\left(  A;B\right)  _{\omega^{x}}.
\]
From this, we conclude that for any mixed input state $\sigma_{AA^{\prime}}$,
the following inequality holds%
\[
E_{\operatorname{sq}}\left(  A;B\right)  _{\omega}\leq\max_{\left\vert
\phi\right\rangle _{AA^{\prime}}}E_{\operatorname{sq}}\left(  A;B\right)
_{\rho},
\]
where $\rho_{AB}\equiv\mathcal{N}_{A^{\prime}\rightarrow B}(\left\vert
\phi\right\rangle \left\langle \phi\right\vert _{AA^{\prime}})$, so that it
suffices to optimize over pure bipartite input states.
\end{remark}

\begin{remark}
Note that we can indeed take a maximization (rather than a supremization) over
pure bipartite inputs if the input space is finite-dimensional because in this
case, the input space is compact and the squashed entanglement measure is
continuous \cite{AF04}.
\end{remark}

\begin{lemma}
\label{lem:alt-char-SE}We can alternatively write the squashed entanglement of
a quantum channel as%
\[
\frac{1}{2}\max_{\rho_{A^{\prime}}}\inf_{V_{E\rightarrow E^{\prime}F}}\left[
H\left(  B|E^{\prime}\right)  _{\omega}+H\left(  B|F\right)  _{\omega}\right]
,
\]
where the maximization is over density operators $\rho_{A^{\prime}}$ on the
input system $A^{\prime}$, the infimization is over \textquotedblleft
squashing isometries\textquotedblright\ $V_{E\rightarrow E^{\prime}F}$, and
the entropies are with respect to the state $\omega_{BE^{\prime}F}$, defined
as%
\[
\omega_{BE^{\prime}F}\equiv V_{E\rightarrow E^{\prime}F}\left(  U_{A^{\prime
}\rightarrow BE}^{\mathcal{N}}\left(  \rho_{A^{\prime}}\right)  \right)  ,
\]
with $U_{A^{\prime}\rightarrow BE}^{\mathcal{N}}$ an isometric extension of
the channel $\mathcal{N}_{A^{\prime}\rightarrow B}$.
\end{lemma}

\begin{proof}
We prove this simply by manipulating the definition in
(\ref{eq:squashed-ent-channel}). Consider a particular pure state $\left\vert
\phi\right\rangle _{AA^{\prime}}$ and a squashing channel $\mathcal{S}%
_{E\rightarrow E^{\prime}}$. Let $\rho_{A^{\prime}}=\ $Tr$_{A}\left\{
\phi_{AA^{\prime}}\right\}  $ and let $V_{E\rightarrow E^{\prime}F}$ be an
isometric extension of the squashing channel $\mathcal{S}_{E\rightarrow
E^{\prime}}$. Let $\omega_{ABE^{\prime}F}\equiv V_{E\rightarrow E^{\prime}%
F}\left(  U_{A^{\prime}\rightarrow BE}^{\mathcal{N}}\left(  \phi_{AA^{\prime}%
}\right)  \right)  $. Then%
\begin{align*}
I\left(  A;B|E^{\prime}\right)  _{\omega}  &  =H\left(  B|E^{\prime}\right)
_{\omega}-H\left(  B|E^{\prime}A\right)  _{\omega}\\
&  =H\left(  B|E^{\prime}\right)  _{\omega}+H\left(  B|F\right)  _{\omega},
\end{align*}
where the first equality is an identity and the second follows from duality of
conditional entropy (i.e., $H\left(  K|L\right)  =-H\left(  K|M\right)  $ for
any pure tripartite state $\psi_{KLM}$). The statement of the lemma then holds
because the above equality holds for any state $\phi_{AA^{\prime}}$ and any
squashing channel $\mathcal{S}_{E\rightarrow E^{\prime}}$.
\end{proof}

\subsection{Concavity in the input density operator}

\begin{lemma}
\label{lem:SE-concave}The squashed entanglement is concave in the input
density operator $\rho_{A^{\prime}}$. That is, the following function (from
Lemma~\ref{lem:alt-char-SE}) is concave as a function of $\rho_{A^{\prime}}$:%
\[
\frac{1}{2}\inf_{V_{E\rightarrow E^{\prime}F}}\left[  H\left(  B|E^{\prime
}\right)  _{\omega}+H\left(  B|F\right)  _{\omega}\right]  ,
\]
where $\omega_{BE^{\prime}F}\equiv V_{E\rightarrow E^{\prime}F}\left(
U_{A^{\prime}\rightarrow BE}^{\mathcal{N}}\left(  \rho_{A^{\prime}}\right)
\right)  $.
\end{lemma}

\begin{proof}
Let $\rho_{A^{\prime}}=\sum_{x}p_{X}\left(  x\right)  \rho_{A^{\prime}}^{x}$
and let%
\[
\omega_{XBE^{\prime}F}\equiv\sum_{x}p_{X}\left(  x\right)  \left\vert
x\right\rangle \left\langle x\right\vert _{X}\otimes V_{E\rightarrow
E^{\prime}F}\left(  U_{A^{\prime}\rightarrow BE}^{\mathcal{N}}\left(
\rho_{A^{\prime}}^{x}\right)  \right)  .
\]
Let $\omega_{BE^{\prime}F}^{x}\equiv V_{E\rightarrow E^{\prime}F}\left(
U_{A^{\prime}\rightarrow BE}^{\mathcal{N}}\left(  \rho_{A^{\prime}}%
^{x}\right)  \right)  $. Then the statement of the lemma is equivalent to%
\[
\frac{1}{2}\inf_{V_{E\rightarrow E^{\prime}F}}\left[  H\left(  B|E^{\prime
}\right)  _{\omega}+H\left(  B|F\right)  _{\omega}\right]  \geq\frac{1}{2}%
\sum_{x}p_{X}\left(  x\right)  \inf_{V_{E\rightarrow E^{\prime}F}^{x}}\left[
H\left(  B|E^{\prime}\right)  _{\tau^{x}}+H\left(  B|F\right)  _{\tau^{x}%
}\right]  ,
\]
where $\tau_{BE^{\prime}F}^{x}\equiv V_{E\rightarrow E^{\prime}F}^{x}\left(
U_{A^{\prime}\rightarrow BE}^{\mathcal{N}}\left(  \rho_{A^{\prime}}%
^{x}\right)  \right)  $. This follows from concavity of conditional entropy.
That is, consider any state $\omega_{BE^{\prime}F}$ with fixed $\rho
_{A^{\prime}}$ and fixed $U_{A^{\prime}\rightarrow BE}^{\mathcal{N}}$.
Consider the following chain of inequalities:%
\begin{align*}
H\left(  B|E^{\prime}\right)  _{\omega}+H\left(  B|F\right)  _{\omega} &  \geq
H\left(  B|E^{\prime}X\right)  _{\omega}+H\left(  B|FX\right)  _{\omega}\\
&  =\sum_{x}p_{X}\left(  x\right)  \left[  H\left(  B|E^{\prime}\right)
_{\omega^{x}}+H\left(  B|F\right)  _{\omega^{x}}\right]  \\
&  \geq\sum_{x}p_{X}\left(  x\right)  \inf_{V_{E\rightarrow E^{\prime}F}^{x}%
}\left[  H\left(  B|E^{\prime}\right)  _{\tau^{x}}+H\left(  B|F\right)
_{\tau^{x}}\right]  .
\end{align*}
The first inequality follows from \textquotedblleft conditioning cannot
increase entropy\textquotedblright\ (i.e., $H(K|L)\geq H\left(  K|LM\right)  $
for any state on systems $KLM$). The equality is just a rewriting of the
entropies. The last inequality follows merely by taking an infimum over all
squashing isometries corresponding to the individual states $U_{A^{\prime
}\rightarrow BE}^{\mathcal{N}}\left(  \rho_{A^{\prime}}^{x}\right)  $. We can
then conclude the statement of the lemma since the calculation is independent
of which squashing isometry $V_{E\rightarrow E^{\prime}F}$ we begin with
(i.e., it holds for the infimum).
\end{proof}

With almost the same proof (excluding the last inequality above), we obtain
the following:

\begin{corollary}
\label{cor:SE-concave-isometry-fixed}For a fixed squashing isometry
$V_{E\rightarrow E^{\prime}F}$, the following function is concave in the input
density operator $\rho_{A^{\prime}}$:%
\[
\frac{1}{2}\left[  H\left(  B|E^{\prime}\right)  _{\omega}+H\left(
B|F\right)  _{\omega}\right]  ,
\]
where $\omega_{BE^{\prime}F}\equiv V_{E\rightarrow E^{\prime}F}\left(
U_{A^{\prime}\rightarrow BE}^{\mathcal{N}}\left(  \rho_{A^{\prime}}\right)
\right)  $.
\end{corollary}

\subsection{Subadditivity inequality}

We now provide a statement and proof of the new subadditivity inequality:

\begin{theorem}
\label{thm:subadd}For any five-party pure state $\psi_{AB_{1}E_{1}B_{2}E_{2}}%
$, the following subadditivity inequality holds%
\begin{equation}
E_{\operatorname{sq}}\left(  A;B_{1}B_{2}\right)  _{\psi} \leq
E_{\operatorname{sq}}\left(  AB_{2}E_{2};B_{1}\right)  _{\psi} +
E_{\operatorname{sq}}\left(  AB_{1}E_{1};B_{2}\right)  _{\psi}.
\label{eq:subadd}%
\end{equation}

\end{theorem}

\begin{proof}
Let%
\begin{align*}
\tau_{AB_{1}E_{1}^{\prime}B_{2}E_{2}}  &  \equiv\mathcal{S}_{E_{1}\rightarrow
E_{1}^{\prime}}(\psi_{AB_{1}E_{1}B_{2}E_{2}}),\\
\sigma_{AB_{1}E_{1}B_{2}E_{2}^{\prime}}  &  \equiv\mathcal{S}_{E_{2}%
\rightarrow E_{2}^{\prime}}(\psi_{AB_{1}E_{1}B_{2}E_{2}}),\\
\omega_{AB_{1}E_{1}^{\prime}B_{2}E_{2}^{\prime}}  &  \equiv(\mathcal{S}%
_{E_{1}\rightarrow E_{1}^{\prime}}\otimes\mathcal{S}_{E_{2}\rightarrow
E_{2}^{\prime}})(\psi_{AB_{1}E_{1}B_{2}E_{2}}),
\end{align*}
where each $\mathcal{S}_{E_{i}\rightarrow E_{i}^{\prime}}$ is an arbitrary
local squashing channel. Let $\left\vert \phi^{\omega}\right\rangle
_{AB_{1}E_{1}^{\prime}B_{2}E_{2}^{\prime}R}$ be a purification of $\omega$
with purifying system $R$. The inequality in (\ref{eq:subadd}) is a
consequence of the following chain of inequalities:%
\begin{align*}
2E_{\text{sq}}\left(  A;B_{1}B_{2}\right)  _{\psi}  &  \leq I\left(
A;B_{1}B_{2}|E_{1}^{\prime}E_{2}^{\prime}\right)  _{\omega}\\
&  =H\left(  B_{1}B_{2}|E_{1}^{\prime}E_{2}^{\prime}\right)  _{\omega
}-H\left(  B_{1}B_{2}|E_{1}^{\prime}E_{2}^{\prime}A\right)  _{\omega}\\
&  =H\left(  B_{1}B_{2}|E_{1}^{\prime}E_{2}^{\prime}\right)  _{\phi}+H\left(
B_{1}B_{2}|R\right)  _{\phi}\\
&  \leq H\left(  B_{1}|E_{1}^{\prime}\right)  _{\phi}+H\left(  B_{2}%
|E_{2}^{\prime}\right)  _{\phi}+H\left(  B_{1}|R\right)  _{\phi}+H\left(
B_{2}|R\right)  _{\phi}\\
&  =H\left(  B_{1}|E_{1}^{\prime}\right)  _{\omega}-H\left(  B_{1}|AB_{2}%
E_{1}^{\prime}E_{2}^{\prime}\right)  _{\omega}+H\left(  B_{2}|E_{2}^{\prime
}\right)  _{\omega}-H\left(  B_{2}|AB_{1}E_{1}^{\prime}E_{2}^{\prime}\right)
_{\omega}\\
&  =I\left(  AB_{2}E_{2}^{\prime};B_{1}|E_{1}^{\prime}\right)  _{\omega
}+I\left(  AB_{1}E_{1}^{\prime};B_{2}|E_{2}^{\prime}\right)  _{\omega}\\
&  \leq I\left(  AB_{2}E_{2};B_{1}|E_{1}^{\prime}\right)  _{\tau}+I\left(
AB_{1}E_{1};B_{2}|E_{2}^{\prime}\right)  _{\sigma}.
\end{align*}
The first inequality follows from the definition in
(\ref{eq:squashed-ent-state}). The first equality is a rewriting of the
conditional mutual information. The second equality exploits duality of
conditional entropy: for any pure tripartite state on systems $KLM$, the
equality $H\left(  K|L\right)  +H\left(  K|M\right)  =0$ holds. The second
inequality results from several applications of strong subadditivity (SSA) of
quantum entropy (SSA is the statement that $I\left(  K;L|M\right)  \geq0$ for
an arbitrary state on systems $KLM$) \cite{LR73}. The third equality again
exploits duality of conditional entropy and the last equality is just a
rewriting in terms of conditional mutual informations. The final inequality is
a result of a quantum data processing inequality for conditional mutual
information (see the proof of Proposition~3\ of \cite{CW04}). Since the
calculation above is independent of the choice of the maps $\mathcal{S}%
_{E_{i}\rightarrow E_{i}^{\prime}}$, the system $E_{1}$ purifies the state on
$AB_{1}B_{2}E_{2}$, and the system $E_{2}$ purifies the state on $AB_{1}%
B_{2}E_{1}$, the subadditivity inequality in the statement of the theorem follows.
\end{proof}

\subsection{Additivity}

As a simple corollary, we find that the squashed entanglement of a quantum
channel is an additive, and thus well behaved, function of quantum channels.

\begin{corollary}
\label{cor:additivity}For any two quantum channels $\mathcal{N}$ and
$\mathcal{M}$, the following additivity relation holds%
\[
E_{\operatorname{sq}}\left(  \mathcal{N}\otimes\mathcal{M}\right)
=E_{\operatorname{sq}}\left(  \mathcal{N}\right)  +E_{\operatorname{sq}%
}\left(  \mathcal{M}\right)  .
\]

\end{corollary}

\begin{proof}
First, note that the following inequality holds
\[
E_{\text{sq}}\left(  \mathcal{N}\otimes\mathcal{M}\right)  \geq E_{\text{sq}%
}\left(  \mathcal{N}\right)  +E_{\text{sq}}\left(  \mathcal{M}\right)  ,
\]
because the optimization of $E_{\text{sq}}\left(  \mathcal{N}\otimes
\mathcal{M}\right)  $ includes tensor-product input choices as a special case
and the squashed entanglement is additive for tensor-product states
\cite{CW04}, so that $E_{\text{sq}}\left(  \mathcal{N}\otimes\mathcal{M}%
\right)  $ can only be larger than the sum of the individual terms.

The other inequality%
\begin{equation}
E_{\text{sq}}\left(  \mathcal{N}\otimes\mathcal{M}\right)  \leq E_{\text{sq}%
}\left(  \mathcal{N}\right)  +E_{\text{sq}}\left(  \mathcal{M}\right)
\label{eq:channel-subadd}%
\end{equation}
follows from Theorem~\ref{thm:subadd}. Indeed, let $\left\vert \phi
\right\rangle _{AA_{1}A_{2}}$ denote any pure-state input to the tensor
product channel $\mathcal{N}\otimes\mathcal{M}$, so that the output state is
$\left(  \mathcal{N}_{A_{1}\rightarrow B_{1}}\otimes\mathcal{M}_{A_{2}%
\rightarrow B_{2}}\right)  (\left\vert \phi\right\rangle \left\langle
\phi\right\vert _{AA_{1}A_{2}})$. Let $U_{A_{1}\rightarrow B_{1}E_{1}%
}^{\mathcal{N}}$ be an isometric extension of $\mathcal{N}_{A_{1}\rightarrow
B_{1}}$ and let $V_{A_{2}\rightarrow B_{2}E_{2}}^{\mathcal{M}}$ be an
isometric extension of $\mathcal{M}_{A_{2}\rightarrow B_{2}}$. Define the
following states:%
\begin{align*}
\left\vert \psi\right\rangle _{AB_{1}E_{1}A_{2}}  &  \equiv U_{A_{1}%
\rightarrow B_{1}E_{1}}^{\mathcal{N}}\left\vert \phi\right\rangle
_{AA_{1}A_{2}},\\
\left\vert \chi\right\rangle _{AA_{1}B_{2}E_{2}}  &  \equiv V_{A_{2}%
\rightarrow B_{2}E_{2}}^{\mathcal{M}}\left\vert \phi\right\rangle
_{AA_{1}A_{2}},\\
\left\vert \varphi\right\rangle _{AB_{1}E_{1}B_{2}E_{2}}  &  \equiv
U_{A_{1}\rightarrow B_{1}E_{1}}^{\mathcal{N}}\otimes V_{A_{2}\rightarrow
B_{2}E_{2}}^{\mathcal{M}}\left\vert \phi\right\rangle _{AA_{1}A_{2}}.
\end{align*}
For any input state $\left\vert \phi\right\rangle _{AA_{1}A_{2}}$, the
following holds%
\begin{align*}
E_{\text{sq}}\left(  A;B_{1}B_{2}\right)  _{\varphi}  &  \leq E_{\text{sq}%
}\left(  AB_{2}E_{2};B_{1}\right)  _{\varphi}+E_{\text{sq}}\left(  AB_{1}%
E_{1};B_{2}\right)  _{\varphi}\\
&  =E_{\text{sq}}\left(  AA_{2};B_{1}\right)  _{\psi}+E_{\text{sq}}\left(
AA_{1};B_{2}\right)  _{\chi}\\
&  \leq E_{\text{sq}}\left(  \mathcal{N}\right)  +E_{\text{sq}}\left(
\mathcal{M}\right)  .
\end{align*}
The first inequality is an application of Theorem~\ref{thm:subadd}. The
equality follows because squashed entanglement is invariant under local
isometries \cite{CW04}. The final inequality follows because $\left\vert
\phi\right\rangle _{AA_{1}A_{2}}$ is a particular pure-state input to the
channel $\mathcal{N}_{A_{1}\rightarrow B_{1}}$ (with $A_{2}$ being the input
and $AA_{1}$ being the purifying system), so that $E_{\text{sq}}\left(
AA_{2};B_{1}\right)  _{\psi}\leq E_{\text{sq}}\left(  \mathcal{N}\right)  $,
and a similar observation for the inequality $E_{\text{sq}}\left(
AA_{1};B_{2}\right)  _{\chi}\leq E_{\text{sq}}\left(  \mathcal{M}\right)  $.
Since the calculation is independent of which pure state $\left\vert
\phi\right\rangle _{AA_{1}A_{2}}$ we begin with, the inequality in
(\ref{eq:channel-subadd}) follows.
\end{proof}

\section{Upper bound on capacities assisted by unlimited forward and backward
communication}

The squashed entanglement of a quantum channel finds it main application in
the theorems given in this section.

\begin{theorem}
\label{thm:SE-upp-bnd}$E_{\operatorname{sq}}\left(  \mathcal{N}\right)  $ is
an upper bound on $Q_{2}\left(  \mathcal{N}\right)  $, the quantum capacity of
a channel $\mathcal{N}$\ assisted by unlimited forward and backward classical
communication:%
\[
Q_{2}\left(  \mathcal{N}\right)  \leq E_{\operatorname{sq}}\left(
\mathcal{N}\right)  .
\]

\end{theorem}

\begin{proof}
First recall that the squashed entanglement is monotone under local operations
and classical communication (LOCC), in the sense that $E_{\text{sq}}\left(
A;B\right)  _{\rho}\geq E_{\text{sq}}\left(  A;B\right)  _{\sigma}$ if Alice
and Bob can obtain the state $\sigma_{AB}$ from $\rho_{AB}$ by LOCC
\cite{CW04}. Furthermore, the squashed entanglement is normalized \cite{CW04},
in the sense that $E_{\text{sq}}\left(  A;B\right)  _{\Phi}=\log d$ for a
maximally entangled state defined as%
\[
\left\vert \Phi\right\rangle _{AB}\equiv\frac{1}{\sqrt{d}}\sum_{i}\left\vert
i\right\rangle _{A}\left\vert i\right\rangle _{B},
\]
where $\{|i\rangle_{A}\}$ and $\{|i\rangle_{B}\}$ are complete orthonormal
bases for quantum systems $A$ and $B$, respectively. Finally, the squashed
entanglement satisfies the following continuity inequality \cite{AF04,C06}:%
\[
\text{if \ \ \ }\left\Vert \rho_{AB}-\sigma_{AB}\right\Vert _{1}%
\leq\varepsilon,\text{ \ \ \ then\ \ \ \ }\left\vert E_{\text{sq}}\left(
A;B\right)  _{\rho}-E_{\text{sq}}\left(  A;B\right)  _{\sigma}\right\vert
\leq16\sqrt{\varepsilon}\log d+4h_{2}\left(  2\sqrt{\varepsilon}\right)  ,
\]
where $d=\min\left\{  \left\vert A\right\vert ,\left\vert B\right\vert
\right\}  $ and $h_{2}\left(  x\right)  $ is the binary entropy function with
the property that $\lim_{x\rightarrow0}h_{2}\left(  x\right)  =0$. The most
general $\left(  n,R,\varepsilon\right)  $\ protocol in this setting begins
with Alice preparing a state $\rho_{AA_{1}\cdots A_{n}}^{\left(  1\right)  }$
on $n+1$ systems. She then transmits the system $A_{1}$ through one use of the
channel $\mathcal{N}$, and considering its isometric extension $U_{A_{1}%
\rightarrow B_{1}E_{1}}^{\mathcal{N}}$, we write the output state as
$\sigma_{AB_{1}E_{1}A_{2}\cdots A_{n}}^{\left(  1\right)  }$. Let $R^{(1)}$ be
a system that purifies this state. There is then a round of an arbitrary
amount of LOCC\ between Alice and Bob, resulting in a state $\rho_{AB_{1}%
E_{1}A_{2}\cdots A_{n}}^{\left(  2\right)  }$. This procedure continues, with
Alice transmitting system $A_{2}$ through the channel, leading to a state
$\sigma_{AB_{1}E_{1}B_{2}E_{2}A_{3}\cdots A_{n}}^{\left(  2\right)  }$, etc.
After the $n$th channel use, the state is $\sigma_{AB_{1}E_{1}B_{2}E_{2}\cdots
B_{n}E_{n}}^{\left(  n\right)  }$ (note that the dimension of the system $A$
might change throughout the protocol). Let $R^{(n)}$ be a system that purifies
this state. There is a final round of LOCC, producing a state $\omega
_{ABE_{1}\cdots E_{n}}$, whose reduction $\omega_{AB}$ satisfies%
\[
\left\Vert \omega_{AB}-\left\vert \Phi\right\rangle \left\langle
\Phi\right\vert _{AB}\right\Vert _{1}\leq\varepsilon,
\]
where $\left\vert \Phi\right\rangle _{AB}$ is the maximally entangled state
with Schmidt rank $2^{nR}$.\ We can now proceed by bounding the entanglement
generation rate of any such protocol as follows:%
\begin{align*}
nR  &  =E_{\text{sq}}\left(  A;B\right)  _{\Phi}\\
&  \leq E_{\text{sq}}\left(  A;B\right)  _{\omega}+nf\left(  \varepsilon
\right)  .
\end{align*}
The equality follows from the normalization of the squashed entanglement on
maximally entangled states (as mentioned above). The inequality follows from
continuity of squashed entanglement with an appropriate choice of $f\left(
\varepsilon\right)  $ so that $\lim_{\varepsilon\rightarrow0}f\left(
\varepsilon\right)  =0$. Continuing,
\begin{align*}
E_{\text{sq}}\left(  A;B\right)  _{\omega}  &  \leq E_{\text{sq}}\left(
A;B_{1}\cdots B_{n}\right)  _{\sigma^{\left(  n\right)  }}\\
&  \leq E_{\text{sq}}(AB_{1}E_{1}\cdots B_{n-1}E_{n-1}R^{(n)};B_{n}%
)_{\sigma^{\left(  n\right)  }}\\
&  \ \ \ \ \ +E_{\text{sq}}\left(  AB_{n}E_{n};B_{1}\cdots B_{n-1}\right)
_{\sigma^{\left(  n\right)  }}\\
&  \leq E_{\text{sq}}\left(  \mathcal{N}\right)  +E_{\text{sq}}\left(
AB_{n}E_{n};B_{1}\cdots B_{n-1}\right)  _{\sigma^{\left(  n\right)  }}\\
&  =E_{\text{sq}}\left(  \mathcal{N}\right)  +E_{\text{sq}}\left(
AA_{n};B_{1}\cdots B_{n-1}\right)  _{\rho^{\left(  n\right)  }}\\
&  \leq nE_{\text{sq}}\left(  \mathcal{N}\right)  .
\end{align*}
The first inequality follows from monotonicity of the squashed entanglement
under LOCC. The second inequality is an application of the subadditivity
inequality in Theorem~\ref{thm:subadd}. The third inequality follows because
$E_{\text{sq}}(AB_{1}E_{1}\cdots B_{n-1}E_{n-1}R^{(n)};B_{n})_{\sigma^{\left(
n\right)  }}\leq E_{\text{sq}}\left(  \mathcal{N}\right)  $ (there is a
particular input to the $n$th channel, while the systems $AB_{1}E_{1}\cdots
B_{n-1}E_{n-1}R^{(n)}$ purify the system being input to the channel). The sole
equality follows because the squashed entanglement is invariant under local
isometries (the isometry here being the isometric extension of the channel).
The last inequality follows by induction, i.e., repeating this procedure by
using monotonicity under LOCC and subadditivity, \textquotedblleft peeling
off\textquotedblright\ one term at a time. Putting everything together, we
arrive at
\[
nR\leq nE_{\text{sq}}\left(  \mathcal{N}\right)  +nf\left(  \varepsilon
\right)  ,
\]
which we can divide by $n$ and take the limit as $\varepsilon\rightarrow0$ to
recover the result that $Q_{2}\left(  \mathcal{N}\right)  \leq E_{\text{sq}%
}\left(  \mathcal{N}\right)  $.
\end{proof}

\begin{remark}
Observe that, in spite of the fact that the squashed entanglement of a quantum
channel is difficult to compute exactly, it is useful in obtaining upper
bounds on the assisted capacities $Q_{2}\left(  \mathcal{N}\right)  $ and
$P_{2}\left(  \mathcal{N}\right)  $ (see Theorem~\ref{thm:SE-upp-bnd-priv}%
\ below for $P_{2}\left(  \mathcal{N}\right)  $) because any squashing channel
leads to an upper bound.
\end{remark}

\begin{remark}
Just as the squashed entanglement of a channel $\mathcal{N}$\ serves as an
upper bound on $Q_{2}\left(  \mathcal{N}\right)  $, we can in fact find other
single-letter upper bounds on $Q_{2}\left(  \mathcal{N}\right)  $ from any
function on quantum states that satisfies LOCC\ monotonicity, asymptotic
continuity, normalization (it equals $\log d$ for a maximally entangled state
of dimension $d$), invariance under local unitaries, and the subadditivity
inequality in Theorem~\ref{thm:subadd}. This follows because these were the
only properties of $E_{\operatorname{sq}}\left(  \mathcal{N}\right)  $ that we
used to prove the above theorem.
\end{remark}

A variation of this setting is one in which there is a forward quantum channel
$\mathcal{N}$\ connecting Alice to Bob and a backward quantum channel
$\mathcal{M}$\ connecting Bob to Alice. The most general protocol for
communicating quantum data (or equivalently in this setting, generating
entanglement) has Alice and Bob each prepare a state on $n$ systems, Alice
sends one system through the forward channel, they conduct a round of LOCC,
Bob sends one of his systems through the backward channel, they conduct a
round of LOCC, etc. By essentially the same proof technique as above, it
follows that $E_{\text{sq}}\left(  \mathcal{N}\right)  +E_{\text{sq}}\left(
\mathcal{M}\right)  $ is an upper bound on the total rate of quantum
communication they can generate with these channels.

\begin{theorem}
\label{thm:SE-upp-bnd-priv}The squashed entanglement $E_{\operatorname{sq}%
}\left(  \mathcal{N}\right)  $ serves as an upper bound on the secret-key
agreement capacity $P_{2}\left(  \mathcal{N}\right)  $\ of a quantum channel
$\mathcal{N}$:%
\[
P_{2}\left(  \mathcal{N}\right)  \leq E_{\operatorname{sq}}\left(
\mathcal{N}\right)  .
\]

\end{theorem}

\begin{proof}
Christandl \textit{et al}.~showed in prior work that the squashed entanglement
is a \textit{secrecy monotone}, in the sense that it does not increase under
local operations and public classical (LOPC) communication \cite{C06,CEHHOR07}%
. The method for doing so was to exploit the fact that LOPC distillation of
secret key is equivalent to LOCC distillation of private states
\cite{HHHO05,HHHO09}. Combining this with the fact that squashed entanglement
is normalized, in the sense that it is never smaller than $k$ for a $k$-bit
private state (see Proposition~4.19 of \cite{C06}), and a proof essentially
identical to the proof of Theorem~\ref{thm:SE-upp-bnd}, we recover that
$P_{2}\left(  \mathcal{N}\right)  \leq E_{\text{sq}}\left(  \mathcal{N}%
\right)  $.
\end{proof}

By similar arguments as above, the secret-key agreement capacity is upper
bounded by $E_{\text{sq}}\left(  \mathcal{N}\right)  +E_{\text{sq}}\left(
\mathcal{M}\right)  $ in the setting where there is a forward quantum channel
$\mathcal{N}$ and a backward quantum channel$~\mathcal{M}$.

\begin{remark}
Just as the squashed entanglement of a channel $\mathcal{N}$\ serves as an
upper bound on $P_{2}\left(  \mathcal{N}\right)  $, we can find other
single-letter upper bounds on it from any function on quantum states that
satisfies LOCC\ monotonicity, asymptotic continuity, normalization (it is
never smaller than $k$ for a $k$-bit private state), invariance under local
unitaries, and the subadditivity inequality in Theorem~\ref{thm:subadd}.
Again, this follows because these were the only properties that we used to
prove the above theorem.
\end{remark}

\begin{remark}
The squashed entanglement of a quantum channel is never larger than its
entanglement cost \cite{BBCW13}. This is an immediate consequence of
Corollary~6 of \cite{CW04}.
\end{remark}

\section{Application to Pauli channels}

In this section, we apply Theorems~\ref{thm:SE-upp-bnd} and
\ref{thm:SE-upp-bnd-priv} to the case of a Pauli channel. That is, we
establish an upper bound on $Q_{2}\left(  \mathcal{P}\right)  $ and
$P_{2}\left(  \mathcal{P}\right)  $ where $\mathcal{P}$ is a Pauli channel,
defined as%
\begin{equation}
\mathcal{P}(\rho)=p_{0}\rho+p_{1}X\rho X+p_{2}Y\rho Y+p_{3}Z\rho Z.
\label{eq:Pauli_channel}%
\end{equation}
In the above, the probabilities $p_{i}$ are non-negative, $\sum_{i}p_{i}=1$,
and%
\begin{equation}
I=\left[
\begin{array}
[c]{cc}%
1 & 0\\
0 & 1
\end{array}
\right]  ,\quad X=\left[
\begin{array}
[c]{cc}%
0 & 1\\
1 & 0
\end{array}
\right]  ,\quad Y=\left[
\begin{array}
[c]{cc}%
0 & -i\\
i & 0
\end{array}
\right]  ,\quad Z=\left[
\begin{array}
[c]{cc}%
1 & 0\\
0 & -1
\end{array}
\right]  , \label{eq:Pauli_operators}%
\end{equation}
are the Pauli operators. We also denote the Pauli operators by $\sigma^{0}$,
\ldots, $\sigma^{3}$, respectively. Note that the Pauli channel is equivalent
to%
\[
\mathcal{P}(\rho)=p_{0}\rho+p_{1}X\rho X+p_{2}XZ\rho ZX+p_{3}Z\rho Z,
\]
due to the fact that $Y=iXZ$ and $XZ=-ZX$.

\begin{theorem}
\label{thm:Pauli-bound}The squashed entanglement leads to the following upper
bound on $Q_{2}\left(  \mathcal{P}\right)  $ and $P_{2}\left(  \mathcal{P}%
\right)  $:%
\[
Q_{2}\left(  \mathcal{P}\right)  ,\ P_{2}\left(  \mathcal{P}\right)  \leq
\min_{\varphi_{1},\varphi_{2},\varphi_{3}}\frac{1}{2}\left[  H(\lambda
)+H(\lambda^{\prime})\right]  -1,
\]
where $\mathcal{P}$ is a Pauli channel and $H(\lambda)$ is the Shannon entropy
of the distribution $\lambda=\{\lambda_{0},\lambda_{1},\lambda_{2},\lambda
_{3}\}$, with%
\begin{align}
\lambda_{0} &  =\frac{1}{4}\left\vert \sqrt{p_{0}}+e^{i\varphi_{3}}\sqrt
{p_{3}}+e^{i\varphi_{1}}\sqrt{p_{1}}-e^{i\varphi_{2}}\sqrt{p_{2}}\right\vert
^{2},\\
\lambda_{1} &  =\frac{1}{4}\left\vert \sqrt{p_{0}}+e^{i\varphi_{3}}\sqrt
{p_{3}}-e^{i\varphi_{1}}\sqrt{p_{1}}+e^{i\varphi_{2}}\sqrt{p_{2}}\right\vert
^{2},\\
\lambda_{2} &  =\frac{1}{4}\left\vert \sqrt{p_{0}}-e^{i\varphi_{3}}\sqrt
{p_{3}}+e^{i\varphi_{1}}\sqrt{p_{1}}+e^{i\varphi_{2}}\sqrt{p_{2}}\right\vert
^{2},\\
\lambda_{3} &  =\frac{1}{4}\left\vert -\sqrt{p_{0}}+e^{i\varphi_{3}}%
\sqrt{p_{3}}+e^{i\varphi_{1}}\sqrt{p_{1}}+e^{i\varphi_{2}}\sqrt{p_{2}%
}\right\vert ^{2},
\end{align}
and $\lambda^{\prime}$ is the same as $\lambda$ except with the substitution
$\varphi_{2}\rightarrow\varphi_{2}+\pi$.
\end{theorem}

\begin{proof}
Let%
\begin{equation}
|\Phi^{\pm}\rangle=\frac{1}{\sqrt{2}}(|00\rangle\pm|11\rangle),\quad|\Psi
^{\pm}\rangle=\frac{1}{\sqrt{2}}(|01\rangle\pm|10\rangle),
\label{eq:Bell_states}%
\end{equation}
denote the four Bell states.

An isometric extension of the Pauli channel is as follows:\ the environment
prepares a Bell state $|\Phi^{+}\rangle_{EF}$ and interacts qubit $E$ with the
input state (in system $A^{\prime}$) according to a controlled Pauli gate.
That is, this isometric extension $W_{A^{\prime}\rightarrow BEF}$\ of
$\mathcal{P}_{A^{\prime}\rightarrow B}$ acting on an input $\left\vert
\psi\right\rangle _{A^{\prime}}$ is as follows:%
\begin{multline}
W_{A^{\prime}\rightarrow BEF}\left\vert \psi\right\rangle _{A^{\prime}}%
=\sqrt{p_{0}}\left\vert \psi\right\rangle _{B}\left\vert \Phi^{+}\right\rangle
_{EF}+\sqrt{p_{1}}e^{i\varphi_{1}}X\left\vert \psi\right\rangle _{B}\left\vert
\Psi^{+}\right\rangle _{EF}\label{eq:after_Pauli_gate}\\
+\sqrt{p_{2}}e^{i\varphi_{2}}XZ\left\vert \psi\right\rangle _{B}\left\vert
\Psi^{-}\right\rangle _{EF}+e^{i\varphi_{3}}Z\left\vert \psi\right\rangle
_{B}\left\vert \Phi^{-}\right\rangle _{EF},
\end{multline}
where $\varphi_{1,2,3}$ are free parameters that we can choose later (we
already take $\varphi_{0}=0$ because invariance of quantum states under a
global phase eliminates one of these degrees of freedom). Note that tracing
out systems $E$ and $F$ from (\ref{eq:after_Pauli_gate}) gives
(\ref{eq:Pauli_channel}) with $\left\vert \psi\right\rangle $\ as input
(meaning that (\ref{eq:after_Pauli_gate}) is a legitimate isometric
extension). Let the squashing channel consist of tracing over system$~F$.

We now argue that if the channel's isometry and the corresponding squashing
channel are fixed to be as above, then the optimal input state on systems
$AA^{\prime}$ to maximize the conditional mutual information $I(A;B|E)$ is the
maximally entangled state $\left\vert \Phi^{+}\right\rangle _{AA^{\prime}}$.
Indeed, consider that the Pauli channel is covariant, so that for all Pauli
operators $U$ there exists a unitary $V$ such that%
\[
\mathcal{P}\left(  U\left(  \psi\right)  \right)  =V\left(  \mathcal{P}\left(
\psi\right)  \right)  .
\]
Let $\mathcal{P}^{1}\left(  \cdot\right)  =$Tr$_{F}\left\{  W\left(
\cdot\right)  W^{\dag}\right\}  $ and $\mathcal{P}^{2}\left(  \cdot\right)
=$Tr$_{E}\left\{  W\left(  \cdot\right)  W^{\dag}\right\}  $. Both of these
channels are covariant, in the sense that for all Pauli operators $U$, there
exist unitaries $V_{B}^{1}$ and $V_{E}^{1}$ such that%
\[
\mathcal{P}^{1}\left(  U\left(  \psi\right)  \right)  =\left(  V_{B}%
^{1}\otimes V_{E}^{1}\right)  \left(  \mathcal{P}\left(  \psi\right)  \right)
\left(  V_{B}^{1}\otimes V_{E}^{1}\right)  ^{\dag}.
\]
Similarly, for all Pauli operators $U$, there exist unitaries $V_{B}^{2}$ and
$V_{F}^{2}$ such that%
\[
\mathcal{P}^{2}\left(  U\left(  \psi\right)  \right)  =\left(  V_{B}%
^{2}\otimes V_{F}^{2}\right)  \left(  \mathcal{P}\left(  \psi\right)  \right)
\left(  V_{B}^{2}\otimes V_{F}^{2}\right)  ^{\dag}.
\]
Equivalently, by inspecting (\ref{eq:after_Pauli_gate}), we see that%
\begin{align*}
W_{A^{\prime}\rightarrow BEF}X_{A^{\prime}}\left\vert \psi\right\rangle
_{A^{\prime}}  & =\left(  X_{B}\otimes X_{E}\otimes X_{F}\right)
W_{A^{\prime}\rightarrow BEF}\left\vert \psi\right\rangle _{A^{\prime}},\\
W_{A^{\prime}\rightarrow BEF}Z_{A^{\prime}}\left\vert \psi\right\rangle
_{A^{\prime}}  & =\left(  Z_{B}\otimes Z_{E}\otimes Z_{F}\right)
W_{A^{\prime}\rightarrow BEF}\left\vert \psi\right\rangle _{A^{\prime}},
\end{align*}
and from this observation (extended by linearity), the covariance stated above
follows. This means that%
\[
H\left(  B|E\right)  _{\omega}+H\left(  B|F\right)  _{\omega}=H\left(
B|E\right)  _{\tau}+H\left(  B|F\right)  _{\tau},
\]
where $\omega$ is the state resulting from preparing a state $\rho_{A^{\prime
}}$ at the input $A^{\prime}$ and $\tau$ is the state resulting from preparing
$U\rho_{A^{\prime}}U^{\dag}$, with $U$ any Pauli operator. We can then apply
Corollary~\ref{cor:SE-concave-isometry-fixed}\ (concavity of $H\left(
B|E\right)  +H\left(  B|F\right)  $ in the input density operator) to conclude
that the maximizing input density operator is the maximally mixed state, since
$I/2=\frac{1}{4}\sum_{i=0}^{3}\sigma^{i}\rho\sigma^{i}$ for any input
state$~\rho$. Since the maximally mixed state on $A^{\prime}$ is purified by
the maximally entangled state $\left\vert \Phi^{+}\right\rangle _{AA^{\prime}%
}$, we conclude that the maximally entangled state maximizes $I(A;B|E)$
whenever the channel isometry and squashing channel are fixed to be of the
form in (\ref{eq:after_Pauli_gate}) (if the channel isometry is not of this
form, then we can take the squashing channel to consist of a preliminary
isometric rotation to make it have the above form, followed by a tracing out
of system$~F$).

We now evaluate the squashed entanglement upper bound using the above
squashing channel. For this purpose, we recall Lemma~\ref{lem:alt-char-SE},
which expresses the squashed entanglement as%
\[
\tfrac{1}{2} \left[H\left(  B|E\right)  +H\left(  B|F\right)  \right].
\]
So we need to compute the eigenvalues of various reduced density matrices in
order to evaluate the above entropies. To derive the reduced density operator
$\rho_{BE}$, we factorize the systems $A$ and $F$ of the state $W_{A^{\prime
}\rightarrow BEF}\left\vert \Phi^{+}\right\rangle _{AA^{\prime}}$ as follows:%
\begin{align}
&  \sqrt{p_{0}}|\Phi_{+}\rangle_{AB}|\Phi_{+}\rangle_{EF}+e^{i\varphi_{3}%
}\sqrt{p_{3}}|\Phi_{-}\rangle_{AB}|\Phi_{-}\rangle_{EF}+e^{i\varphi_{1}}%
\sqrt{p_{1}}|\Psi_{+}\rangle_{AB}|\Psi_{+}\rangle_{EF}+e^{i\varphi_{2}}%
\sqrt{p_{2}}|\Psi_{-}\rangle_{AB}|\Psi_{-}\rangle_{EF},\label{eq:factorize_AF}%
\\
&  =\frac{1}{\sqrt{2}}|0\rangle_{A}\left(  \sqrt{p_{0}}|0\rangle_{B}|\Phi
_{+}\rangle_{EF}+e^{i\varphi_{3}}\sqrt{p_{3}}|0\rangle_{B}|\Phi_{-}%
\rangle_{EF}+e^{i\varphi_{1}}\sqrt{p_{1}}|1\rangle_{B}|\Psi_{+}\rangle
_{EF}+e^{i\varphi_{2}}\sqrt{p_{2}}|1\rangle_{B}|\Psi_{-}\rangle_{EF}\right)
\nonumber\\
&  \quad+\frac{1}{\sqrt{2}}|1\rangle_{A}\left(  \sqrt{p_{0}}|1\rangle_{B}%
|\Phi_{+}\rangle_{EF}-e^{i\varphi_{3}}\sqrt{p_{3}}|1\rangle_{B}|\Phi
_{-}\rangle_{EF}+e^{i\varphi_{1}}\sqrt{p_{1}}|0\rangle_{B}|\Psi_{+}%
\rangle_{EF}-e^{i\varphi_{2}}\sqrt{p_{2}}|0\rangle_{B}|\Psi_{-}\rangle
_{EF}\right) \nonumber\\
&  =\frac{1}{2}|0\rangle_{A}|0\rangle_{F}\left(  \sqrt{p_{0}}|00\rangle
+e^{i\varphi_{3}}\sqrt{p_{3}}|00\rangle+e^{i\varphi_{1}}\sqrt{p_{1}}%
|11\rangle-e^{i\varphi_{2}}\sqrt{p_{2}}|11\rangle\right)  _{BE}\nonumber\\
&  \quad+\frac{1}{2}|0\rangle_{A}|1\rangle_{F}\left(  \sqrt{p_{0}}%
|01\rangle-e^{i\varphi_{3}}\sqrt{p_{3}}|01\rangle+e^{i\varphi_{1}}\sqrt{p_{1}%
}|10\rangle+e^{i\varphi_{2}}\sqrt{p_{2}}|10\rangle\right)  _{BE}\nonumber\\
&  \quad+\frac{1}{2}|1\rangle_{A}|0\rangle_{F}\left(  \sqrt{p_{0}}%
|10\rangle-e^{i\varphi_{3}}\sqrt{p_{3}}|10\rangle+e^{i\varphi_{1}}\sqrt{p_{1}%
}|01\rangle+e^{i\varphi_{2}}\sqrt{p_{2}}|01\rangle\right)  _{BE}\nonumber\\
&  \quad+\frac{1}{2}|1\rangle_{A}|1\rangle_{F}\left(  \sqrt{p_{0}}%
|11\rangle+e^{i\varphi_{3}}\sqrt{p_{3}}|11\rangle+e^{i\varphi_{1}}\sqrt{p_{1}%
}|00\rangle-e^{i\varphi_{2}}\sqrt{p_{2}}|00\rangle\right)  _{BE}.\nonumber
\end{align}
Tracing out $A$ and $F$, we have%
\begin{multline}
\rho_{BE}=\frac{1}{4}\left[  \left\{  (\sqrt{p_{0}}+e^{i\varphi_{3}}%
\sqrt{p_{3}})|00\rangle+(e^{i\varphi_{1}}\sqrt{p_{1}}-e^{i\varphi_{2}}%
\sqrt{p_{2}})|11\rangle\right\}  \left\{  h.c.\right\}  \right. \nonumber\\
+\left\{  (\sqrt{p_{0}}+e^{i\varphi_{3}}\sqrt{p_{3}})|11\rangle+(e^{i\varphi
_{1}}\sqrt{p_{1}}-e^{i\varphi_{2}}\sqrt{p_{2}})|00\rangle\right\}  \left\{
h.c.\right\} \nonumber\\
+\left\{  (\sqrt{p_{0}}-e^{i\varphi_{3}}\sqrt{p_{3}})|01\rangle+(e^{i\varphi
_{1}}\sqrt{p_{1}}+e^{i\varphi_{2}}\sqrt{p_{2}})|10\rangle\right\}  \left\{
h.c.\right\} \nonumber\\
\left.  +\left\{  (\sqrt{p_{0}}-e^{i\varphi_{3}}\sqrt{p_{3}})|10\rangle
+(e^{i\varphi_{1}}\sqrt{p_{1}}+e^{i\varphi_{2}}\sqrt{p_{2}})|01\rangle
\right\}  \left\{  h.c.\right\}  \right]  ,
\end{multline}
which consists of two block diagonalized submatrices in the subspaces spanned
by $\{|00\rangle,\,|11\rangle\}$ and $\{|01\rangle,\,|10\rangle\}$,
respectively:%
\[
=\frac{1}{4}\left[
\begin{array}
[c]{cc}%
|a|^{2}+|b|^{2} & ab^{\ast}+a^{\ast}b\\
ab^{\ast}+a^{\ast}b & |a|^{2}+|b|^{2}%
\end{array}
\right]  \oplus\frac{1}{4}\left[
\begin{array}
[c]{cc}%
|c|^{2}+|d|^{2} & cd^{\ast}+c^{\ast}d\\
cd^{\ast}+c^{\ast}d & |c|^{2}+|d|^{2}%
\end{array}
\right]  ,
\]
where%
\begin{align}
a  &  =\sqrt{p_{0}}+e^{i\varphi_{3}}\sqrt{p_{3}},\nonumber\\
b  &  =e^{i\varphi_{1}}\sqrt{p_{1}}-e^{i\varphi_{2}}\sqrt{p_{2}},\nonumber\\
c  &  =\sqrt{p_{0}}-e^{i\varphi_{3}}\sqrt{p_{3}},\nonumber\\
d  &  =e^{i\varphi_{1}}\sqrt{p_{1}}+e^{i\varphi_{2}}\sqrt{p_{2}}.\nonumber
\end{align}
This is diagonalized by the following unitary transformation:%
\[
U=\frac{1}{\sqrt{2}}\left[
\begin{array}
[c]{cc}%
1 & 1\\
1 & -1
\end{array}
\right]  \oplus\frac{1}{\sqrt{2}}\left[
\begin{array}
[c]{cc}%
1 & 1\\
1 & -1
\end{array}
\right]
\]
so that%
\[
U\rho_{BE}U^{\dagger}=\frac{1}{4}\left[
\begin{array}
[c]{cc}%
|a+b|^{2} & 0\\
0 & |a-b|^{2}%
\end{array}
\right]  \oplus\frac{1}{4}\left[
\begin{array}
[c]{cc}%
|c+d|^{2} & 0\\
0 & |c-d|^{2}%
\end{array}
\right]  .
\]
Thus we find the following four eigenvalues for $\rho_{BE}$:%
\begin{align}
\lambda_{0}  &  =\frac{1}{4}\left\vert \sqrt{p_{0}}+e^{i\varphi_{3}}%
\sqrt{p_{3}}+e^{i\varphi_{1}}\sqrt{p_{1}}-e^{i\varphi_{2}}\sqrt{p_{2}%
}\right\vert ^{2},\label{eq:rho_BE_eigenvalues}\\
\lambda_{1}  &  =\frac{1}{4}\left\vert \sqrt{p_{0}}+e^{i\varphi_{3}}%
\sqrt{p_{3}}-e^{i\varphi_{1}}\sqrt{p_{1}}+e^{i\varphi_{2}}\sqrt{p_{2}%
}\right\vert ^{2},\\
\lambda_{2}  &  =\frac{1}{4}\left\vert \sqrt{p_{0}}-e^{i\varphi_{3}}%
\sqrt{p_{3}}+e^{i\varphi_{1}}\sqrt{p_{1}}+e^{i\varphi_{2}}\sqrt{p_{2}%
}\right\vert ^{2},\\
\lambda_{3}  &  =\frac{1}{4}\left\vert -\sqrt{p_{0}}+e^{i\varphi_{3}}%
\sqrt{p_{3}}+e^{i\varphi_{1}}\sqrt{p_{1}}+e^{i\varphi_{2}}\sqrt{p_{2}%
}\right\vert ^{2},
\end{align}
from which we can calculate the von Neumann entropy as%
\begin{equation}
H(BE)_{\rho}=H(\lambda), \label{eq:rho_BE_entropy}%
\end{equation}
where $H(\lambda)$ is the Shannon entropy of the distribution $\lambda
=\{\lambda_{0},\lambda_{1},\lambda_{2},\lambda_{3}\}$.

Tracing over the $B$ system results in the maximally mixed state on the $E$
system, so that%
\[
H\left(  E\right)  =1.
\]
Similarly, we find that the reduced state on system $F$ is maximally mixed, so
that $H\left(  F\right)  =1$. Now, if we instead trace over systems $A$ and
$E$, the calculation of the eigenvalues of the reduced density matrix on
systems $B$ and $F$ is similar to that detailed above. However, observe that
all of the Bell states are invariant under a swap, with the exception of
$\left\vert \Psi^{-}\right\rangle $. So, starting from (\ref{eq:factorize_AF}%
), we realize that the eigenvalues are the same, except we have the
substitution $\varphi_{2}\rightarrow\varphi_{2}+\pi$ due to the previous
observation. We then recover the statement of the theorem.
\end{proof}

As a lower bound on both $Q_{2}\left(  \mathcal{P}\right)  $ and$\ P_{2}%
\left(  \mathcal{P}\right)  $, both the direct and the reverse coherent
information \cite{GPLS09} for a Pauli channel with the Bell state input are
given by%
\begin{equation}
I_{R}(\mathcal{N}_{P},\Phi^{+})=H(A)_{\rho}-H(AB)_{\rho}=\max\left\{
0,1-H(\mathbf{p})\right\}  , \label{eq:reverse_coh_info_Pauli}%
\end{equation}
where $\mathbf{p}=\{p_{0},p_{1},p_{2},p_{3}\}$.

\subsection{Dephasing channel}

The dephasing channel $\mathcal{D}$\ is obtained from a Pauli channel by
setting $p_{1}=p_{2}=0$. The eigenvalues of $\rho_{BE}$ in this case reduce to%
\begin{align}
\lambda_{0,1} &  =\frac{1}{4}\left(  1+2\sqrt{p(1-p)}\cos\varphi_{3}\right)
,\label{eq:eigenvalues_bit_flip}\\
\lambda_{2,3} &  =\frac{1}{4}\left(  1-2\sqrt{p(1-p)}\cos\varphi_{3}\right)  ,
\end{align}
where we set $p=p_{0}$. Clearly, $H(\lambda)=H\left(  \lambda^{\prime}\right)
$ and $H(\lambda)$ is minimized by setting $\varphi_{3}=0$. So our upper bound
becomes%
\[
Q_{2}\left(  \mathcal{D}\right)  ,\ P_{2}\left(  \mathcal{D}\right)  \leq
h_{2}\left(  \frac{1+2\sqrt{p\left(  1-p\right)  }}{2}\right)  ,
\]
where $h_{2}\left(  \cdot\right)  $ is the binary entropy function. Note that
the above bound is equal to the entanglement cost of the dephasing channel
(compare with (83) of \cite{BBCW13}). Both the direct and the reverse coherent
information for a dephasing channel are given by $I_{R}=1-h_{2}(p)$.%
%TCIMACRO{\FRAME{ftbpFU}{3.4904in}{2.3333in}{0pt}{\Qcb{Upper and lower bounds
%on $Q_{2}$ and $P_{2}$ for a qubit dephasing channel. The dashed curve is both
%the direct and reverse coherent information, while the solid curve is the
%squashed entanglement upper bound. Note that this latter curve is equal to the
%entanglement cost from \cite{BBCW13} for this particular channel.}%
%}{\Qlb{fig:dephasing}}{dephasing_channel.pdf}%
%width 3.4904in;  height 2.3333in;  depth 0pt;  original-width 6.4446in;
%original-height 4.2921in;  cropleft "0";  croptop "1";  cropright "1";
%cropbottom "0";  filename 'dephasing_channel.pdf';file-properties "XNPEU";}}}%
%BeginExpansion
\begin{figure}
[ptb]
\begin{center}
\includegraphics[
natheight=4.292100in,
natwidth=6.444600in,
height=2.3333in,
width=3.4904in
]%
{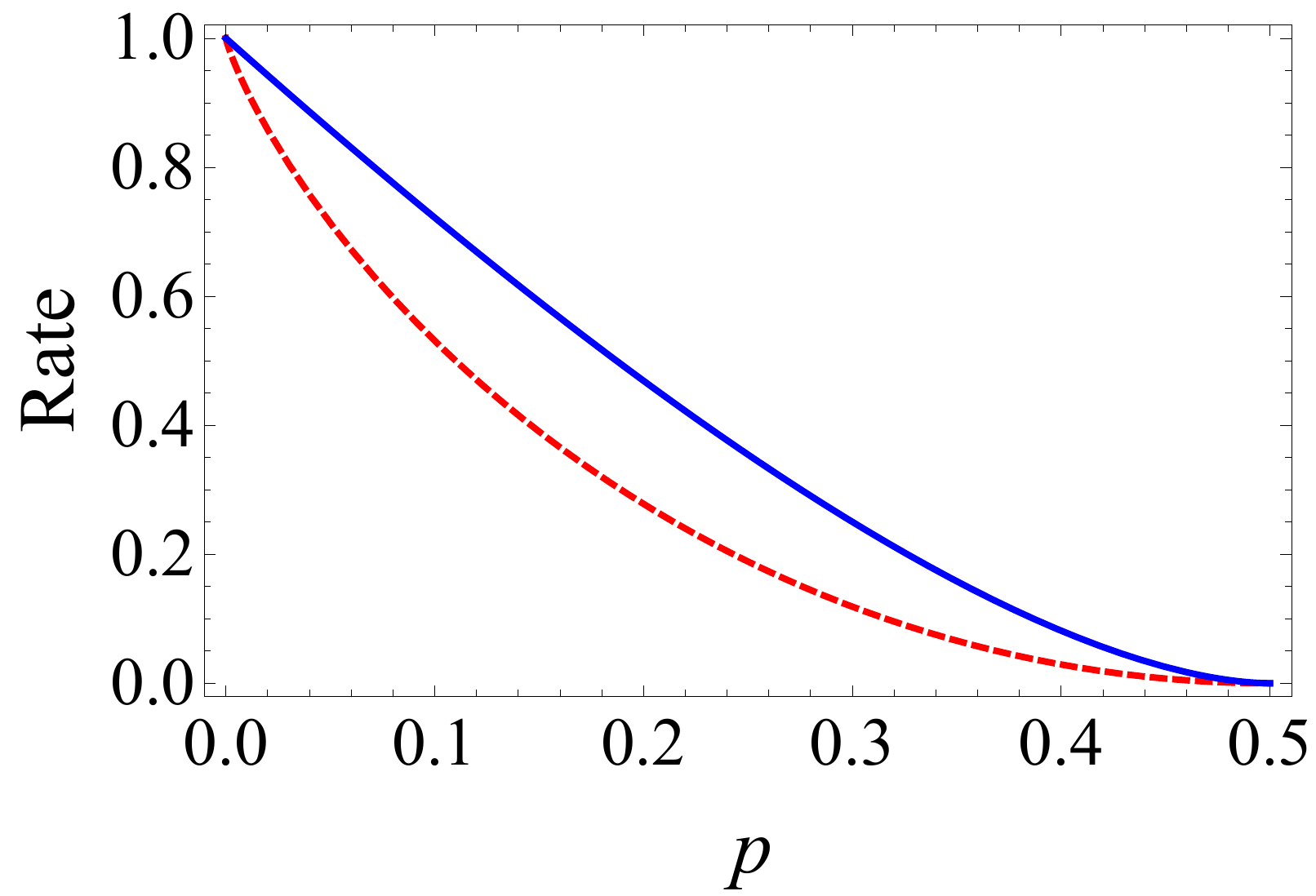}%
\caption{Upper and lower bounds on $Q_{2}$ and $P_{2}$ for a qubit dephasing
channel. The dashed curve is both the direct and reverse coherent information,
while the solid curve is the squashed entanglement upper bound. Note that this
latter curve is equal to the entanglement cost from \cite{BBCW13} for this
particular channel.}%
\label{fig:dephasing}%
\end{center}
\end{figure}
%EndExpansion

\subsection{Depolarizing channel}

The depolarizing channel is also a special case of a Pauli channel:%
\begin{align}
\mathcal{N}_{\text{dep}}\left(  \rho\right)   &  =(1-p)\rho+p\frac{I}%
{2}\nonumber\label{eq:depolarizing_channel}\\
&  =\left(  1-\frac{3p}{4}\right)  \rho+\frac{p}{4}\left(  X\rho X+Y\rho
Y+Z\rho Z\right)  .
\end{align}
Numerical work indicates that the minimizing choice for the phases
$\varphi_{1}$, $\varphi_{2}$, and $\varphi_{3}$ from
Theorem~\ref{thm:Pauli-bound}\ is simply $\varphi_{1}=\varphi_{2}=\varphi
_{3}=0$. Figure~\ref{fig:depolarizing}\ plots the squashed entanglement upper
bound and both the direct and the reverse coherent information for this
channel. This figure makes it clear that our squashed entanglement upper bound
is not particularly tight in this case because the qubit depolarizing channel
becomes entanglement-breaking whenever $p\geq2/3$ and thus $Q_{2}%
(\mathcal{N}_{\text{dep}})=0$ whenever $p\geq2/3$ \cite{R01}.%
%TCIMACRO{\FRAME{ftbpFU}{3.5111in}{2.3869in}{0pt}{\Qcb{Upper and lower bounds
%on $Q_{2}$ and $P_{2}$ for a qubit depolarizing channel. The dashed curve is
%both the direct and reverse coherent information, while the solid curve is the
%squashed entanglement upper bound.}}{\Qlb{fig:depolarizing}}%
%{depolarizing_channel.pdf}{\special{ language "Scientific Word";
%type "GRAPHIC";  maintain-aspect-ratio TRUE;  display "USEDEF";
%valid_file "F";  width 3.5111in;  height 2.3869in;  depth 0pt;
%original-width 6.4584in;  original-height 4.3751in;  cropleft "0";
%filename 'depolarizing_channel.pdf';file-properties "XNPEU";}}}%
%BeginExpansion
\begin{figure}
[ptb]
\begin{center}
\includegraphics[
natheight=4.375100in,
natwidth=6.458400in,
height=2.3869in,
width=3.5111in
]%
{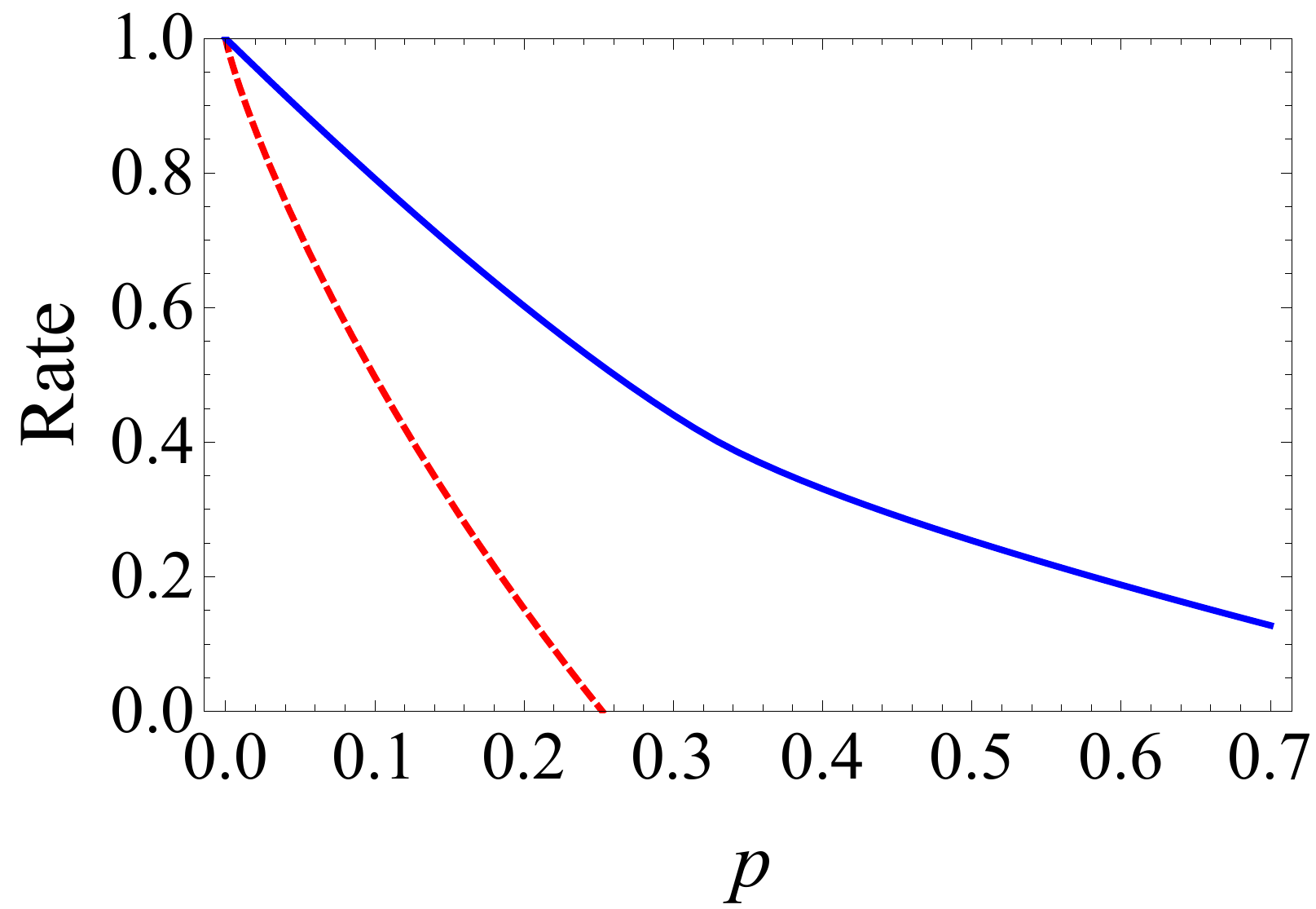}%
\caption{Upper and lower bounds on $Q_{2}$ and $P_{2}$ for a qubit
depolarizing channel. The dashed curve is both the direct and reverse coherent
information, while the solid curve is the squashed entanglement upper bound.}%
\label{fig:depolarizing}%
\end{center}
\end{figure}
%EndExpansion

\section{Application to bosonic channels}

\subsection{Pure-loss channel}

As a final contribution, we apply our bounds to the practically relevant
pure-loss bosonic channel, which is modeled by the following Heisenberg
picture evolution:%
\begin{equation}
\hat{b}=\sqrt{\eta}\hat{a}+\sqrt{1-\eta}\hat{e},
\label{eq:pure-loss-Heisenberg}%
\end{equation}
where $\hat{a}$, $\hat{b}$, and $\hat{e}$ are the electromagnetic field mode
operators corresponding to the sender's input, the receiver's output, and the
environmental input. For the pure-loss bosonic channel, the environment
injects the vacuum state. The parameter $\eta\in\left[  0,1\right]  $
characterizes the transmissivity of the channel, i.e., the fraction of input
photons that make it to the output on average. Let $\mathcal{N}_{\eta}$ denote
the channel to the receiver.

For a classical-communication-assisted quantum communication protocol or
secret-key agreement protocol over such a channel, we assume that it begins
and ends with finite-dimensional states, but the processing between the first
and final step can be with infinite-dimensional systems.\footnote{That is,
their objective is to generate a maximally entangled state $\left\vert
\Phi\right\rangle _{AB}$\ or a finite number of secret key bits, and they do
so by Alice encoding a finite-dimensional quantum state into an
infinite-dimensional system and the final step of the protocol has them
truncate their systems to be of finite dimension. In this way, the continuity
inequality in the proof of Theorem~\ref{thm:SE-upp-bnd} safely applies and all
of the other steps in between involve only the quantum data processing
inequality, which has been proven to hold in the general infinite-dimensional
setting \cite{U77}.} Furthermore, we impose a mean photon number constraint at
the input of each channel, i.e., for each channel input, we have the
constraint that $\left\langle \hat{a}^{\dag}\hat{a}\right\rangle \leq N_{S}$
for some $N_{S}$ such that $0\leq N_{S}<\infty$. Thus, $E_{\text{sq}}\left(
\mathcal{N}_{\eta}\right)  $ with the additional photon number constraint on
the channel input is an upper bound on both $Q_{2}\left(  \mathcal{N}_{\eta
}\right)  $ and $P_{2}\left(  \mathcal{N}_{\eta}\right)  $. By taking the
squashing channel for the environment to be another pure-loss bosonic channel
of transmissivity $\eta_{1}\in\left[  0,1\right]  $, noting that the resulting
conditional mutual information can be written as a sum of two conditional
entropies as in Lemma~\ref{lem:alt-char-SE}, and applying the extremality of
Gaussian states with respect to conditional entropies \cite{EW07,WGC06}, we
find the following upper bounds on $E_{\text{sq}}\left(  \mathcal{N}_{\eta
}\right)  $ for all $\eta_{1}\in\left[  0,1\right]  $ (see
Appendix~\ref{app:first} for a detailed proof):%
\begin{multline}
\tfrac{1}{2}\Big[g\left(  \left(  1-\eta_{1}+\eta\eta_{1}\right)
N_{S}\right)  +g\left(  \left(  \eta_{1}+\eta\left(  1-\eta_{1}\right)
\right)  N_{S}\right)  
-g\left(  \eta_{1}\left(  1-\eta\right)  N_{S}\right)  -g\left(  \left(
1-\eta_{1}\right)  \left(  1-\eta\right)  N_{S}\right)  \Big], \label{eq:bosonic-upper-bounds}
\end{multline}
where $g\left(  x\right)  \equiv\left(  x+1\right)  \log_{2}\left(
x+1\right)  -x\log_{2}x$ is the entropy of a bosonic, circularly-symmetric
thermal state with mean photon number $x$. The function in
(\ref{eq:bosonic-upper-bounds}) is symmetric and convex in $\eta_{1}$ (see
Appendix~\ref{app:convexity}), so that its minimum occurs at $\eta_{1}=1/2$,
leading to the following simpler upper bound:%
\[
g\left(  \left(  1+\eta\right)  N_{S}/2\right)  -g\left(  \left(
1-\eta\right)  N_{S}/2\right)  .
\]
By taking the limit of this upper bound as $N_{S}\rightarrow\infty$, we
recover the following photon-number independent upper bound on the capacities
$Q_{2}(\mathcal{N}_{\eta})$ and $P_{2}(\mathcal{N}_{\eta})$:%
\begin{equation}
\log_{2}\left(  \frac{1+\eta}{1-\eta}\right)  .\label{eq:good-upper-bound}%
\end{equation}
For values of $\eta\ll1$ (which we expect in practical scenarios with high
loss), this upper bound is close to the following lower bound on
$Q_{2}(\mathcal{N}_{\eta})$ and $P_{2}(\mathcal{N}_{\eta})$ established in
\cite{GPLS09,PGBL09}:%
\begin{equation}
\log_{2}\left(  \frac{1}{1-\eta}\right)  .\label{eq:lower-bound-q2-p2}%
\end{equation}
Thus, for such small $\eta$, our upper bound demonstrates that the protocols
from \cite{GPLS09,PGBL09} achieving the lower bound in
(\ref{eq:lower-bound-q2-p2}) are nearly optimal.
Figure~\ref{fig:pure-loss} plots these bounds.
\begin{figure}
[ptb]
\begin{center}
\includegraphics[
%natheight=4.292100in,
%natwidth=6.444600in,
%height=2.3333in,
width=3.4904in
]%
{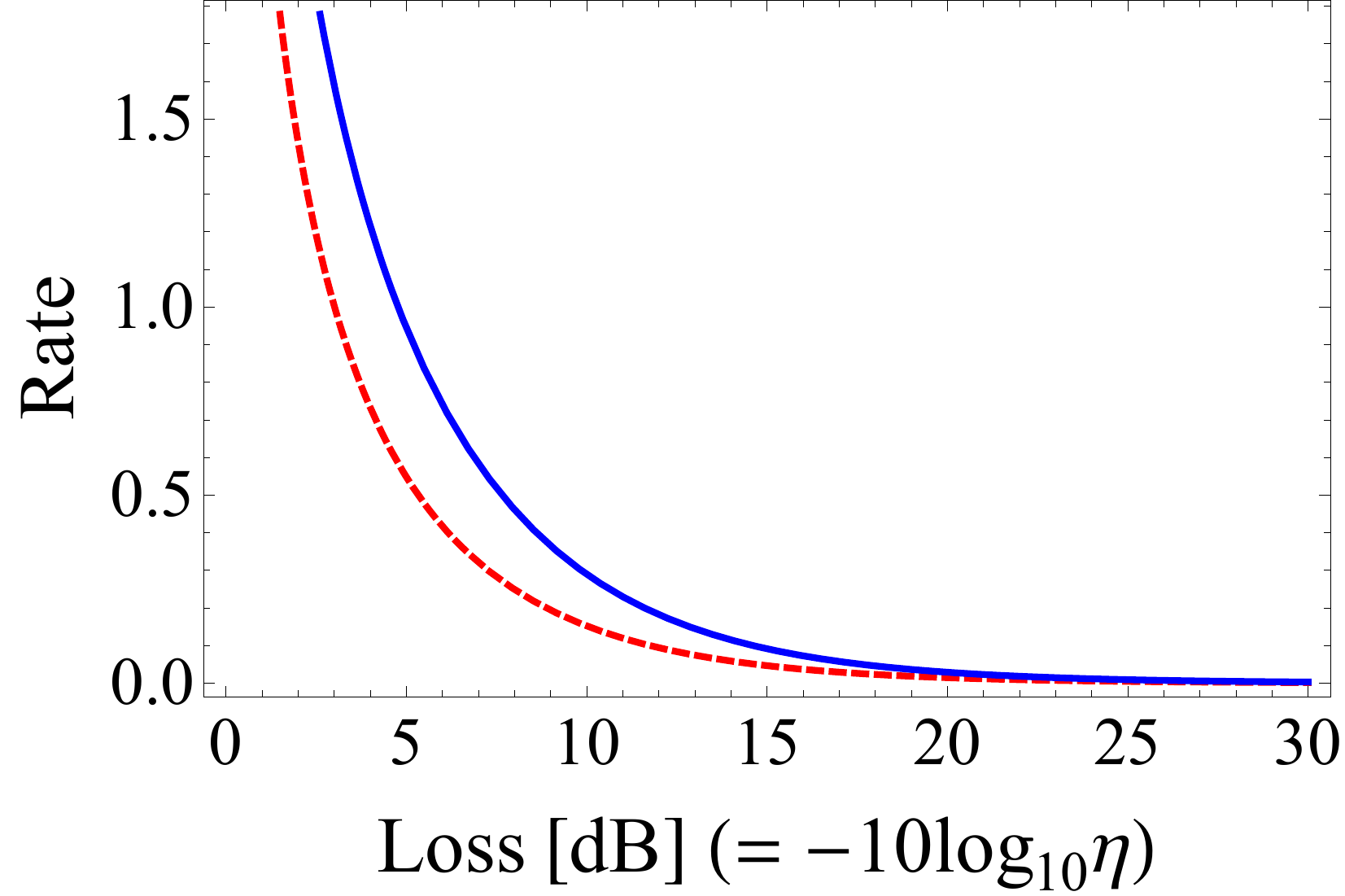}%
\caption{Upper and lower bounds on $Q_{2}$ and $P_{2}$ for a pure-loss
bosonic channel. The dashed curve is both the direct and reverse coherent information,
while the solid curve is the squashed entanglement upper bound.}%
\label{fig:pure-loss}%
\end{center}
\end{figure}

\begin{remark}
In \cite{HolevoWerner2001}, Holevo and Werner proved by a different approach
that $\log_{2}\left(  \left(  1+\eta\right)  /\left(  1-\eta\right)  \right)
$ serves as an upper bound on the unassisted quantum capacity of a pure-loss
bosonic channel with transmissivity $\eta\in\left[  0,1\right]  $. However, it
is not clear to us that their method generalizes to yield an upper bound on
the quantum capacity assisted by unlimited forward and backward classical
communication. Furthermore, in light of later results \cite{WPG07}\ which
established $\log_{2}\left(  \eta/\left(  1-\eta\right)  \right)  $ as an
upper bound on the unassisted quantum capacity, it is clear that the
Holevo-Werner bound is not tight.
\end{remark}

\subsection{Phase-insensitive Gaussian channels}

In this section, we find bounds on $Q_{2}\left(  \mathcal{N}\right)  $ and
$P_{2}\left(  \mathcal{N}\right)  $ whenever $\mathcal{N}$ is a
phase-insensitive Gaussian channel \cite{WPGCRSL12,PhysRevLett.108.110505},
meaning that it adds an equal amount of noise to each quadrature of the
electromagnetic field. Examples of these channels include the pure-loss
channel, the thermal channel, the additive noise channel, and the
phase-insensitive amplifier channel.

Such channels have the following action on the mean vector $x$\ and covariance
matrix $\Gamma$\ of a given single-mode, bosonic quantum state
\cite{WPGCRSL12}:%
\begin{align}
x  &  \rightarrow Kx,\label{eq:PI-channel-1}\\
\Gamma &  \rightarrow K\Gamma K^{T}+N, \label{eq:PI-channel-2}%
\end{align}
where $K$ and $N$ are square matrices satisfying%
\begin{align}
N  &  \geq0,\\
\det N  &  \geq\left(  \det K-1\right)  ^{2},
\end{align}
in order for the map to be a legitimate completely positive and trace
preserving map. A phase-insensitive channel has%
\begin{align}
K  &  =\text{diag}\left(  \sqrt{\tau},\sqrt{\tau}\right)  ,\\
N  &  =\text{diag}\left(  \nu,\nu\right)  , \label{eq:PI-channel-last}%
\end{align}
where $\tau\in\left[  0,1\right]  $ corresponds to attenuation, $\tau\geq1$
amplification, and $\nu$ is the variance of an additive noise.

A powerful (albeit simple) theorem in continuous-variable quantum information
is that any phase-insensitive Gaussian channel as given above can be
decomposed as the concatenation of a pure-loss channel $\mathcal{L}_{T}$\ with
loss parameter $T$ followed by an amplifier channel $\mathcal{A}_{G}$\ with
gain $G$, i.e.,%
\begin{equation}
\mathcal{N=\mathcal{A}}_{G}\circ\mathcal{L}_{T},\label{eq:decomp-gaussian}%
\end{equation}
where $\mathcal{N}$ is given by (\ref{eq:PI-channel-1}%
)-(\ref{eq:PI-channel-last}) and $T$ and $G$ are chosen such that $\tau=TG$
and $\nu=G\left(  1-T\right)  +G-1$ \cite{CGH06,PhysRevLett.108.110505}. These
equations are equivalent to $T=2\tau/\left(  \tau+\nu+1\right)  $ and
$G=\left(  \tau+\nu+1\right)  /2$.

Now consider that $E_{\text{sq}}\left(  \mathcal{N}\right)  =E_{\text{sq}%
}\left(  \mathcal{\mathcal{A}}_{G}\circ\mathcal{L}_{T}\right)  \leq
E_{\text{sq}}\left(  \mathcal{L}_{T}\right)  $, where the inequality follows
from quantum data processing (the quantum conditional mutual information does
not increase under processing of one of the systems that is not the
conditioning system---see the proof of Proposition~3\ of \cite{CW04}).
Combining this fact with the bound for the pure-loss channel from the previous
section, we find the following upper bounds on $Q_{2}\left(  \mathcal{N}%
\right)  $ and $P_{2}\left(  \mathcal{N}\right)  $:%
\begin{equation}
Q_{2}\left(  \mathcal{N}\right)  ,P_{2}\left(  \mathcal{N}\right)  \leq
\log_{2}\left(  \frac{1+T}{1-T}\right)  .
\end{equation}

We can specialize the above result to the case of a thermal channel and an
additive noise channel. The evolution for the thermal channel is the same as
that in (\ref{eq:pure-loss-Heisenberg}), with the exception that the
environment is prepared in a thermal state of mean photon number $N_{B}\geq0$.
The decomposition for the thermal channel then corresponds to that
in\ (\ref{eq:decomp-gaussian}), with%
\begin{align*}
T &  =\frac{\eta}{\left(  1-\eta\right)  N_{B}+1},\\
G &  =\left(  1-\eta\right)  N_{B}+1.
\end{align*}
This is because for the thermal channel, we have $\tau=\eta$ and $\nu=\left(
1-\eta\right)  \left(  2N_{B}+1\right)  $. Thus, we find the following upper
bound on $Q_{2}$ and $P_{2}$ for the thermal channel:%
\begin{equation}
\log\left(  \frac{1+\frac{\eta}{\left(  1-\eta\right)  N_{B}+1}}{1-\frac{\eta
}{\left(  1-\eta\right)  N_{B}+1}}\right)  =\log\left(  \frac{\left(
1-\eta\right)  N_{B}+1+\eta}{\left(  1-\eta\right)  N_{B}+1-\eta}\right)
.\label{eq:thermal-channel-upper-bound}%
\end{equation}
The additive noise channel corresponds to the following map:%
\[
\rho\rightarrow\int d^{2}\alpha\ \frac{1}{\pi\overline{n}}\exp\left\{
-\left\vert \alpha\right\vert ^{2}/\overline{n}\right\}  D\left(
\alpha\right)  \rho D^{\dag}\left(  \alpha\right)  ,
\]
where $D\left(  \alpha\right)  $ is a unitary displacement operator and
$\overline{n}>0$ is the noise variance \cite{WPGCRSL12}. It is well known that
the additive noise channel is equivalent to a thermal channel for which
$N_{B}\rightarrow\infty$ and $\eta\rightarrow1$, while $\left(
1-\eta\right)  N_{B}\rightarrow\overline{n}$ \cite{GGLMS04}. From this
relation, we immediately obtain the following upper bound on both $Q_{2}$ and
$P_{2}$ for an additive noise channel:%
\begin{equation}
\log\left(  \frac{\overline{n}+2}{\overline{n}}\right)  .
\end{equation}

\section{Conclusion}

We have established the squashed entanglement of a quantum channel as an
operationally relevant, well behaved information measure for quantum channels.
Our work here finds application in quantum key distribution, establishing the
first simple upper bound on the rate that is achievable over the pure-loss
bosonic channel, which models pure-loss free-space and fiber-optic
communication. When the environment mode is noisy (for example, in a thermal
state)---which is the case for a quantum key distribution setting for instance
when the eavesdropper makes an active attack---the secret-key rate reduces
from the case of the passive eavesdropper. In this case, the upper bound in
(\ref{eq:thermal-channel-upper-bound}) serves as a general upper bound to the
rate achievable over a repeater-less lossy channel, using any key distribution protocol.

An interesting open question is to establish $E_{\text{sq}}(\mathcal{N})$ as a
strong converse rate for $Q_{2}(\mathcal{N})$ and $P_{2}(\mathcal{N})$
(Theorem~\ref{thm:SE-upp-bnd} establishes $E_{\text{sq}}(\mathcal{N})$ as a
weak converse rate). A strong converse rate is defined to be such that if a
communication scheme exceeds it, then the error probability converges to one
as the number of channel uses becomes large. In this sense, establishing
$E_{\text{sq}}(\mathcal{N})$ as a strong converse rate for $Q_{2}%
(\mathcal{N})$ and $P_{2}(\mathcal{N})$ would significantly enhance
Theorems~\ref{thm:SE-upp-bnd} and \ref{thm:SE-upp-bnd-priv} given here. For
proving this, some combination of the ideas in \cite{BBCW13,O08} might be
helpful. 

\bigskip

\textbf{Acknowledgements}. We are grateful to Francesco Buscemi, Seth Lloyd,
Cosmo Lupo, and Andreas Winter for insightful discussions. We also acknowledge
Mark Byrd, Eric Chitambar, and the other participants of the Boris Musulin
Workshop on Open Quantum Systems and Information Processing for helpful
feedback. Finally, we thank Bob Tucci for kindly pointing us to his related
work on squashed entanglement. This research was supported by the DARPA
Quiness Program through US Army Research Office award W31P4Q-12-1-0019.

\appendix

\section{Squashed entanglement upper bound for the pure-loss bosonic channel}

\label{app:first}Here we detail a proof that (\ref{eq:bosonic-upper-bounds})
is an upper bound on $Q_{2}\left(  \mathcal{N}_{\eta}\right)  $ and
$P_{2}\left(  \mathcal{N}_{\eta}\right)  $, where $\mathcal{N}_{\eta}$ is a
pure-loss bosonic channel with transmissivity $\eta\in\left[  0,1\right]  $.

As mentioned before, we need to consider only pure states $\left\vert
\phi\right\rangle _{AA^{\prime}}$ when optimizing the squashed entanglement of
a quantum channel. Let $U_{E\rightarrow E^{\prime}F}^{\mathcal{S}}$ be an
isometric extension of Eve's squashing channel $\mathcal{S}_{E\rightarrow
E^{\prime}}$. Let $|\psi\rangle_{ABE^{\prime}F}\equiv U_{E\rightarrow
E^{\prime}F}^{\mathcal{S}}U_{A^{\prime}\rightarrow BE}^{\mathcal{N}}\left\vert
\phi\right\rangle _{AA^{\prime}}$, so that, $\mathrm{Tr}_{F}[|\phi
\rangle\langle\phi|_{ABE^{\prime}F}]=\mathcal{S}_{E\rightarrow E^{\prime}%
}\circ U_{A^{\prime}\rightarrow BE}^{\mathcal{N}}(\phi_{AA^{\prime}})$. Then
according to Lemma~\ref{lem:alt-char-SE}%
\begin{equation}
\sup_{\phi_{AA^{\prime}}}E_{\text{sq}}(A;B)_{\mathcal{N}_{A^{\prime
}\rightarrow B}(\phi_{AA^{\prime}})}=\sup_{\phi_{AA^{\prime}}}\frac{1}{2}%
\inf_{\mathcal{S}_{E\rightarrow E^{\prime}}}\left(  H(B|E^{\prime})_{\psi
}+H(B|F)_{\psi}\right)  .\label{eq:squashed_entanglement_pure_input}%
\end{equation}
Now suppose that Alice and Bob are connected by a pure-loss bosonic channel
with transmissivity $\eta$. It is not necessarily an easy task to optimize
Eve's squashing channel $\mathcal{S}$. Instead, we consider a specific
squashing channel: a pure-loss bosonic channel $\mathcal{L}_{\eta_{1}}$ with
transmissivity $\eta_{1}$. As shown in Lemma~\ref{lem:alt-char-SE}, the
squashed entanglement can be written as a sum of two conditional entropies,
each of which is a function of the reduced state Tr$_{A}\left\{
\phi_{AA^{\prime}}\right\}  $ on $A^{\prime}$. Since the overall channel from
$A^{\prime}$ to $BE^{\prime}$ is Gaussian and the overall channel from
$A^{\prime}$ to $BF$ is Gaussian and due to the photon-number constraint at
the input, it follows from the extremality of Gaussian states for conditional
entropy \cite{EW07,WGC06}\ that a thermal state on $A^{\prime}$ of mean photon
number $N_{S}$ maximizes both of these quantities. With this and the fact that
$\phi_{AA^{\prime}}$ is a pure state, we can conclude that the optimal
$\phi_{AA^{\prime}}$ is a two-mode squeezed vacuum (TMSV) state. Let $N_{S}$
be the average photon number of one share of the TMSV. Then the covariance
matrix of the reduced thermal state at $A^{\prime}$ is given by
\[
\gamma^{A^{\prime}}=\left[
\begin{array}
[c]{cc}%
1+2N_{S} & 0\\
0 & 1+2N_{S}%
\end{array}
\right]  .
\]
Note that the covariance matrix is defined such that a vacuum state (or
coherent state) is described by an identity matrix. Therefore a covariance
matrix of the initial state in system $A^{\prime}E^{\prime}F$ is given by
$\gamma^{A^{\prime}}\oplus I^{E^{\prime}}\oplus I^{F}$. The beamsplitting
operations are given by the transformation
\[
\gamma^{A^{\prime}}\oplus I^{E^{\prime}}\oplus I^{F}\rightarrow S_{\eta_{1}%
}S_{\eta}\left(  \gamma^{A^{\prime}}\oplus I^{E^{\prime}}\oplus I^{F}\right)
S_{\eta}^{T}S_{\eta_{1}}^{T},
\]
where
\[
S_{\eta}=\left[
\begin{array}
[c]{ccc}%
\sqrt{\eta} & \sqrt{1-\eta} & 0\\
-\sqrt{1-\eta} & \sqrt{\eta} & 0\\
0 & 0 & 1
\end{array}
\right]  ^{\oplus2},\quad S_{\eta_{1}}=\left[
\begin{array}
[c]{ccc}%
1 & 0 & 0\\
0 & \sqrt{\eta_{1}} & \sqrt{1-\eta_{1}}\\
0 & -\sqrt{1-\eta_{1}} & \sqrt{\eta_{1}}%
\end{array}
\right]  ^{\oplus2},
\]
(the superscript \textquotedblleft$\oplus2$\textquotedblright\ means that the
same matrix is applied to both $x$ and $p$ quadratures. Because of the
symmetry of the state and the beamsplitter operation in phase space, basically
we need to consider only one quadrature.) This transformation is easily
calculated and we get a covariance matrix for the state $\mathrm{Tr}%
_{A}\left\{  |\phi\rangle\langle\phi|_{ABE^{\prime}F}\right\}  $:
\begin{multline*}
S_{\eta_{1}}S_{\eta}\left(  \gamma^{A^{\prime}}\oplus I^{E^{\prime}}\oplus
I^{F}\right)  S_{\eta_{1}}^{T}S_{\eta}^{T}=\\
\left[
\begin{array}
[c]{ccc}%
1+\eta2N_{S} & -\sqrt{\eta(1-\eta)}\sqrt{\eta_{1}}2N_{S} & \sqrt{\eta(1-\eta
)}\sqrt{1-\eta_{1}}2N_{S}\\
-\sqrt{\eta(1-\eta)}\sqrt{\eta_{1}}2N_{S} & 1+(1-\eta)\eta_{1}2N_{S} &
-(1-\eta)\sqrt{\eta_{1}(1-\eta_{1})}2N_{S}\\
\sqrt{\eta(1-\eta)}\sqrt{1-\eta_{1}}2N_{S} & -(1-\eta)\sqrt{\eta_{1}%
(1-\eta_{1})}2N_{S} & 1+(1-\eta)(1-\eta_{1})2N_{S}%
\end{array}
\right]  ^{\oplus2}.
\end{multline*}
It immediately implies a covariance matrix of the marginal state on
$E^{\prime}$:
\[
\gamma_{E^{\prime}}=\left[
\begin{array}
[c]{cc}%
1+(1-\eta)\eta_{1}2N_{S} & 0\\
0 & 1+(1-\eta)\eta_{1}2N_{S}%
\end{array}
\right]  ,
\]
which is the covariance matrix for a thermal state with photon number
$(1-\eta)\eta_{1}N_{S}$. Thus we have
\[
H(E^{\prime})=g\left(  (1-\eta)\eta_{1}N_{S}\right)  ,
\]
where $g(x)=(1+x)\log(1+x)-x\log x$. Similarly, we get
\[
H(F)=g\left(  (1-\eta)(1-\eta_{1})N_{S}\right)  .
\]
The other entropies $H(BE^{\prime})$ and $H(BF)$ are also obtained by
considering the corresponding submatrices and diagonalizing them. Then we can
find
\begin{align}
H(BE^{\prime}) &  =g\left(  \{\eta+(1-\eta)\eta_{1}\}N_{S}\right)  \nonumber\\
H(BF) &  =g\left(  \{\eta+(1-\eta)(1-\eta_{1})\}N_{S}\right)
\end{align}
As a consequence, we obtain the upper bound,%
\begin{align}
Q_{2}(\mathcal{N}_{\eta}) &  \leq\min_{\eta_{1}}\frac{1}{2}\big\{g\left(
\{\eta+(1-\eta)\eta_{1}\}N_{S}\right)  -g\left(  (1-\eta)\eta_{1}N_{S}\right)
\nonumber\\
&  \ \ \ \ \ \ \ \ +g\left(  \{\eta+(1-\eta)(1-\eta_{1})\}N_{S}\right)
-g\left(  (1-\eta)(1-\eta_{1})N_{S}\right)
\big\}\label{eq:lossy_bosonic_upper_bound1}\\
&  =g\left(  (1+\eta)N_{S}/2\right)  -g\left(  (1-\eta)N_{S}/2\right)  .
\end{align}
The minimal value is achieved by $\eta_{1}=1/2$ because the function is
symmetric and convex in $\eta_{1}$ (with convexity checked by computing the
second derivative, see next appendix). The expression $g\left(  (1+\eta
)N_{S}/2\right)  -g\left(  (1-\eta)N_{S}/2\right)  $ converges to $\log
(1+\eta)/(1-\eta)$ as $N_{S}\rightarrow\infty$.

\subsection{Convexity in $\eta_{1}$}

\label{app:convexity}We compute the second derivative of the function in
(\ref{eq:lossy_bosonic_upper_bound1}) in order to establish that it is convex.
The function is%
\[
g\left(  \left(  \eta+\left(  1-\eta\right)  \left(  1-\eta_{1}\right)
\right)  N\right)  +g\left(  \left(  \eta+\left(  1-\eta\right)  \eta
_{1}\right)  N\right)  -g\left(  \eta_{1}\left(  1-\eta\right)  N\right)
-g\left(  \left(  1-\eta_{1}\right)  \left(  1-\eta\right)  N\right)
\]

We now compute the second derivative of each term.

Consider that $g\left(  \left(  \eta+\left(  1-\eta\right)  \left(  1-\eta
_{1}\right)  \right)  N\right)  $ is equal to%
\begin{multline*}
\left(  \left(  \eta+\left(  1-\eta\right)  \left(  1-\eta_{1}\right)
\right)  N+1\right)  \log\left(  \left(  \eta+\left(  1-\eta\right)  \left(
1-\eta_{1}\right)  \right)  N+1\right) \\
-\left(  \left(  \eta+\left(  1-\eta\right)  \left(  1-\eta_{1}\right)
\right)  N\right)  \log\left(  \left(  \eta+\left(  1-\eta\right)  \left(
1-\eta_{1}\right)  \right)  N\right)  .
\end{multline*}
The first derivative of the above with respect to $\eta_{1}$ is given by%
\begin{multline*}
-\left(  1-\eta\right)  N\log\left(  \left(  \eta+\left(  1-\eta\right)
\left(  1-\eta_{1}\right)  \right)  N+1\right)  -\left(  1-\eta\right)  N\\
+\left(  1-\eta\right)  N\log\left(  \left(  \eta+\left(  1-\eta\right)
\left(  1-\eta_{1}\right)  \right)  N\right)  +\left(  1-\eta\right)  N\\
=\left(  1-\eta\right)  N\left[  -\log\left(  \left(  \eta+\left(
1-\eta\right)  \left(  1-\eta_{1}\right)  \right)  N+1\right)  +\log\left(
\left(  \eta+\left(  1-\eta\right)  \left(  1-\eta_{1}\right)  \right)
N\right)  \right]
\end{multline*}
The second derivative with respect to $\eta_{1}$ is then given by%
\begin{align*}
&  \left(  1-\eta\right)  N\left[  \frac{\left(  1-\eta\right)  N}{\left(
\left(  \eta+\left(  1-\eta\right)  \left(  1-\eta_{1}\right)  \right)
N+1\right)  }-\frac{\left(  1-\eta\right)  N}{\left(  \left(  \eta+\left(
1-\eta\right)  \left(  1-\eta_{1}\right)  \right)  N\right)  }\right] \\
&  =-\left[  \left(  1-\eta\right)  N\right]  ^{2}\left[  \frac{1}{\left[
\left(  \eta+\left(  1-\eta\right)  \left(  1-\eta_{1}\right)  \right)
N+1\right]  \ \left[  \left(  \eta+\left(  1-\eta\right)  \left(  1-\eta
_{1}\right)  \right)  N\right]  }\right]  .
\end{align*}

Consider that $g\left(  \left(  \eta+\left(  1-\eta\right)  \eta_{1}\right)
N\right)  $ is equal to%
\[
\left(  \left(  \eta+\left(  1-\eta\right)  \eta_{1}\right)  N+1\right)
\log\left(  \left(  \eta+\left(  1-\eta\right)  \eta_{1}\right)  N+1\right)
-\left(  \left(  \eta+\left(  1-\eta\right)  \eta_{1}\right)  N\right)
\log\left(  \left(  \eta+\left(  1-\eta\right)  \eta_{1}\right)  N\right)  .
\]
The first derivative of the above with respect to $\eta_{1}$ is given by%
\begin{multline*}
\left(  1-\eta\right)  N\log\left(  \left(  \eta+\left(  1-\eta\right)
\eta_{1}\right)  N+1\right)  +\left(  1-\eta\right)  N\\
-\left(  1-\eta\right)  N\log\left(  \left(  \eta+\left(  1-\eta\right)
\eta_{1}\right)  N\right)  -\left(  1-\eta\right)  N\\
=\left(  1-\eta\right)  N\left[  \log\left(  \left(  \eta+\left(
1-\eta\right)  \eta_{1}\right)  N+1\right)  -\log\left(  \left(  \eta+\left(
1-\eta\right)  \eta_{1}\right)  N\right)  \right]
\end{multline*}
The second derivative with respect to $\eta_{1}$ is then given by%
\begin{align*}
&  \left(  1-\eta\right)  N\left[  \frac{\left(  1-\eta\right)  N}{\left(
\left(  \eta+\left(  1-\eta\right)  \eta_{1}\right)  N+1\right)  }%
-\frac{\left(  1-\eta\right)  N}{\left(  \left(  \eta+\left(  1-\eta\right)
\eta_{1}\right)  N\right)  }\right]  \\
&  =-\left[  \left(  1-\eta\right)  N\right]  ^{2}\left[  \frac{1}{\left[
\left(  \eta+\left(  1-\eta\right)  \eta_{1}\right)  N+1\right]  \ \left[
\left(  \eta+\left(  1-\eta\right)  \eta_{1}\right)  N\right]  }\right]  .
\end{align*}

Consider that $g\left(  \left(  1-\eta\right)  \eta_{1}N\right)  $ is equal to%
\[
\left(  \left(  \left(  1-\eta\right)  \eta_{1}\right)  N+1\right)
\log\left(  \left(  \left(  1-\eta\right)  \eta_{1}\right)  N+1\right)
-\left(  \left(  \left(  1-\eta\right)  \eta_{1}\right)  N\right)  \log\left(
\left(  \left(  1-\eta\right)  \eta_{1}\right)  N\right)  .
\]
The first derivative of the above with respect to $\eta_{1}$ is given by%
\begin{multline*}
\left(  1-\eta\right)  N\log\left(  \left(  \left(  1-\eta\right)  \eta
_{1}\right)  N+1\right)  +\left(  1-\eta\right)  N\\
-\left(  1-\eta\right)  N\log\left(  \left(  \left(  1-\eta\right)  \eta
_{1}\right)  N\right)  -\left(  1-\eta\right)  N\\
=\left(  1-\eta\right)  N\left[  \log\left(  \left(  \left(  1-\eta\right)
\eta_{1}\right)  N+1\right)  -\log\left(  \left(  \left(  1-\eta\right)
\eta_{1}\right)  N\right)  \right]
\end{multline*}
The second derivative with respect to $\eta_{1}$ is then given by%
\begin{align*}
&  \left(  1-\eta\right)  N\left[  \frac{\left(  1-\eta\right)  N}{\left(
\left(  \left(  1-\eta\right)  \eta_{1}\right)  N+1\right)  }-\frac{\left(
1-\eta\right)  N}{\left(  \left(  \left(  1-\eta\right)  \eta_{1}\right)
N\right)  }\right]  \\
&  =-\left[  \left(  1-\eta\right)  N\right]  ^{2}\left[  \frac{1}{\left[
\left(  \left(  1-\eta\right)  \eta_{1}\right)  N+1\right]  \ \left[  \left(
\left(  1-\eta\right)  \eta_{1}\right)  N\right]  }\right]  .
\end{align*}

Consider that $g\left(  \left(  \left(  1-\eta\right)  \left(  1-\eta
_{1}\right)  \right)  N\right)  $ is equal to%
\[
\left(  \left(  \left(  1-\eta\right)  \left(  1-\eta_{1}\right)  \right)
N+1\right)  \log\left(  \left(  \left(  1-\eta\right)  \left(  1-\eta
_{1}\right)  \right)  N+1\right)  -\left(  \left(  \left(  1-\eta\right)
\left(  1-\eta_{1}\right)  \right)  N\right)  \log\left(  \left(  \left(
1-\eta\right)  \left(  1-\eta_{1}\right)  \right)  N\right)  .
\]
The first derivative of the above with respect to $\eta_{1}$ is given by%
\begin{multline*}
-\left(  1-\eta\right)  N\log\left(  \left(  \left(  1-\eta\right)  \left(
1-\eta_{1}\right)  \right)  N+1\right)  -\left(  1-\eta\right)  N\\
+\left(  1-\eta\right)  N\log\left(  \left(  \left(  1-\eta\right)  \left(
1-\eta_{1}\right)  \right)  N\right)  +\left(  1-\eta\right)  N\\
=\left(  1-\eta\right)  N\left[  -\log\left(  \left(  \left(  1-\eta\right)
\left(  1-\eta_{1}\right)  \right)  N+1\right)  +\log\left(  \left(  \left(
1-\eta\right)  \left(  1-\eta_{1}\right)  \right)  N\right)  \right]
\end{multline*}
The second derivative with respect to $\eta_{1}$ is then given by%
\begin{align*}
&  \left(  1-\eta\right)  N\left[  \frac{\left(  1-\eta\right)  N}{\left(
\left(  \left(  1-\eta\right)  \left(  1-\eta_{1}\right)  \right)  N+1\right)
}-\frac{\left(  1-\eta\right)  N}{\left(  \left(  \left(  1-\eta\right)
\left(  1-\eta_{1}\right)  \right)  N\right)  }\right]  \\
&  =-\left[  \left(  1-\eta\right)  N\right]  ^{2}\left[  \frac{1}{\left[
\left(  \left(  1-\eta\right)  \left(  1-\eta_{1}\right)  \right)  N+1\right]
\ \left[  \left(  \left(  1-\eta\right)  \left(  1-\eta_{1}\right)  \right)
N\right]  }\right]  .
\end{align*}

So now we just need to determine whether the second derivative is positive:%
\begin{align*}
&  -\left[  \left(  1-\eta\right)  N\right]  ^{2}\left[  \frac{1}{\left[
\left(  \eta+\left(  1-\eta\right)  \left(  1-\eta_{1}\right)  \right)
N+1\right]  \ \left[  \left(  \eta+\left(  1-\eta\right)  \left(  1-\eta
_{1}\right)  \right)  N\right]  }\right] \\
&  -\left[  \left(  1-\eta\right)  N\right]  ^{2}\left[  \frac{1}{\left[
\left(  \eta+\left(  1-\eta\right)  \eta_{1}\right)  N+1\right]  \ \left[
\left(  \eta+\left(  1-\eta\right)  \eta_{1}\right)  N\right]  }\right] \\
&  +\left[  \left(  1-\eta\right)  N\right]  ^{2}\left[  \frac{1}{\left[
\left(  \left(  1-\eta\right)  \eta_{1}\right)  N+1\right]  \ \left[  \left(
\left(  1-\eta\right)  \eta_{1}\right)  N\right]  }\right] \\
&  +\left[  \left(  1-\eta\right)  N\right]  ^{2}\left[  \frac{1}{\left[
\left(  \left(  1-\eta\right)  \left(  1-\eta_{1}\right)  \right)  N+1\right]
\ \left[  \left(  \left(  1-\eta\right)  \left(  1-\eta_{1}\right)  \right)
N\right]  }\right]  .
\end{align*}
This simplifies to%
\begin{align*}
&  \left[  \frac{1}{\left[  \left(  \left(  1-\eta\right)  \eta_{1}\right)
N+1\right]  \ \left[  \left(  \left(  1-\eta\right)  \eta_{1}\right)
N\right]  }\right]  +\\
&  \left[  \frac{1}{\left[  \left(  \left(  1-\eta\right)  \left(  1-\eta
_{1}\right)  \right)  N+1\right]  \ \left[  \left(  \left(  1-\eta\right)
\left(  1-\eta_{1}\right)  \right)  N\right]  }\right] \\
&  \geq\left[  \frac{1}{\left[  \left(  \eta+\left(  1-\eta\right)  \eta
_{1}\right)  N+1\right]  \ \left[  \left(  \eta+\left(  1-\eta\right)
\eta_{1}\right)  N\right]  }\right]  +\\
&  \left[  \frac{1}{\left[  \left(  \eta+\left(  1-\eta\right)  \left(
1-\eta_{1}\right)  \right)  N+1\right]  \ \left[  \left(  \eta+\left(
1-\eta\right)  \left(  1-\eta_{1}\right)  \right)  N\right]  }\right]
\end{align*}
This last inequality is true by inspection because the terms on the RHS\ are
the same as those on the LHS, only with an extra factor of $\eta\geq0$ in the denominator.

\bibliographystyle{alpha}
\bibliography{Ref}

\end{document}